\documentclass[11pt,american,round]{article}
\usepackage[T1]{fontenc}
\usepackage[latin9]{inputenc}
\usepackage[a4paper]{geometry}
\usepackage{babel}
\usepackage{prettyref}
\usepackage{enumitem}
\usepackage{amsmath}
\usepackage{amsthm}
\usepackage{amssymb}
\usepackage{setspace}
\usepackage[authoryear]{natbib}
\usepackage[unicode=true,
 bookmarks=false,
 breaklinks=false,pdfborder={0 0 0},pdfborderstyle={},backref=false,colorlinks=false]
 {hyperref}

\makeatletter
\theoremstyle{plain}
    \ifx\thechapter\undefined
	    \newtheorem{theorem}{\protect\theoremname}
	  \else
      \newtheorem{theorem}{\protect\theoremname}[chapter]
    \fi
\theoremstyle{plain}
    \ifx\thechapter\undefined
      \newtheorem{proposition}{\protect\propositionname}
    \else
      \newtheorem{proposition}{\protect\propositionname}[chapter]
    \fi

\@ifundefined{date}{}{\date{}}
\usepackage{titlesec}
\titleformat*{\section}{\Large\bf\scshape}
\titleformat*{\subsection}{\bf\scshape}
\titlespacing{\subsection}{0pt}{*1.2}{*0.6}

\usepackage[flushmargin]{footmisc}

\usepackage{datetime}
\ddmmyyyydate

\usepackage[flushleft]{threeparttable}
\usepackage{booktabs,multirow,multicol} 

\usepackage{tikz}
\usepackage{tikz-qtree}
\usepackage{multicol} 

\usepackage{fancyhdr}
\theoremstyle{definition}

\usepackage{lipsum}

\usepackage{caption}
\makeatother

\providecommand{\propositionname}{Proposition}
\providecommand{\theoremname}{Theorem}

\begin{document}

\title{Optimal Search and Discovery\thanks{This paper previously was circulated under the title ``Optimal Search
and Awareness Expansion''.}}
\author{Rafael P. Greminger\thanks{I am deeply grateful to my advisors, Jaap  Abbring and Tobias Klein, for their thoughtful guidance and support. I also thank the anonymous reviewing team, Bart Bronnenberg, Nikolaus Schweitzer,  as well as members of the Structural Econometrics Group in Tilburg for excellent comments. Finally, I thank the Netherlands Organisation for Scientific Research (NWO) for financial support through Research Talent grant 406.18.568. Tilburg University, Department of Econometrics \& Operations
Research, \protect\href{mailto:r.p.greminger@tilburguniversity.edu}{r.p.greminger@tilburguniversity.edu}.}}
\date{August 31, 2021 \\ 
		Published version: \textcolor{blue}{\url{https://doi.org/10.1287/mnsc.2021.4085}}}
\maketitle
\begin{abstract}
This paper studies a search problem where a consumer is initially
aware of only a few products. At every point in time, the consumer then
decides between searching among alternatives he is already aware of
and discovering more products. I show that the optimal policy for
this search and discovery problem is fully characterized by tractable
reservation values. Moreover, I prove that a predetermined index fully
specifies the purchase decision of a consumer following the optimal
search policy. Finally, a comparison highlights differences to classical
random and directed search. \vspace{1cm}
\end{abstract}

\section{Introduction}

Consumers typically first need to search for product information before being able to compare alternatives.
The resulting search frictions have received considerable attention
in the literature.\footnote{For example, \citet{Stigler1961,Diamond1971,Burdett1983,Anderson1999,Kuksov2006,Choi2016a,Moraga-Gonzalez2017,Moraga-Gonzalez2017a}
study search frictions in equilibrium models and \citet{Hortacsu2004,Hong2006c,Santos2012,Bronnenberg2016,Chen2016,Zhang2018,Jolivet2019}
study implications of search empirically. } 
Under the rational choice paradigm, the analysis of such limited information settings
relies on optimal search policies that describe how a consumer optimally searches among all available alternatives.
I add to this literature by developing and solving a sequential search
problem that introduces a novel aspect: \emph{limited awareness of available products}. 

To fix ideas, consider a consumer looking to buy a mobile phone. Through
advertising or recommendations from friends, the consumer initially
is aware of a single available phone and has some (but not all) information
on what it offers. Given this basic information, the consumer can
directly gather more detailed information on this alternative, for
example by reading a review online. Besides, there are also phones
available that the consumer is initially not aware of. For these alternatives,
he knows neither of their existence, nor the features
they offer. This precludes the consumer from directly inspecting these
phones. Instead, he first needs to discover and become aware of them,
for example by getting more recommendations from friends or through
a search intermediary. Figure \ref{fig:Extended-search-process} depicts 
a possible choice sequence for this case. 

\begin{figure}[th]
\centering{}%
\begin{minipage}[t]{0.8\columnwidth}%
\hspace{-3em}
\scalebox{1}{
\tikzstyle{level 1}=[level distance=3em]
\begin{tikzpicture}[every tree node/.style={minimum size=1.2em, inner sep=0pt,font=\footnotesize}, every node/.style={minimum size=1.2em, inner sep=0pt,font=\footnotesize},
	   level distance=1.7cm,sibling distance=0.4cm,baseline, label = a,
   	edge from parent path={(\tikzparentnode) -- (\tikzchildnode)}]
\Tree [.\node[circle]{} ; 
		\edge[dashed];
		[.\node[]{inspect Phone 1};]
		[.\node[]{discover};
			\edge[dashed];
			\node[]{inspect Phone 1};
			[.\node[]{inspect Phone 2};
				\edge[dashed];
				\node[]{inspect Phone 1} ;
				\node[]{buy Phone 2}; 
				\edge[dashed];
				\node[]{discover};] 
				\edge[dashed];
				\node[]{discover};] ]  
\end{tikzpicture}
}\caption{{\small{}Example of a choice sequence in the search and discovery
problem. \label{fig:Extended-search-process}}}
\end{minipage}
\end{figure}
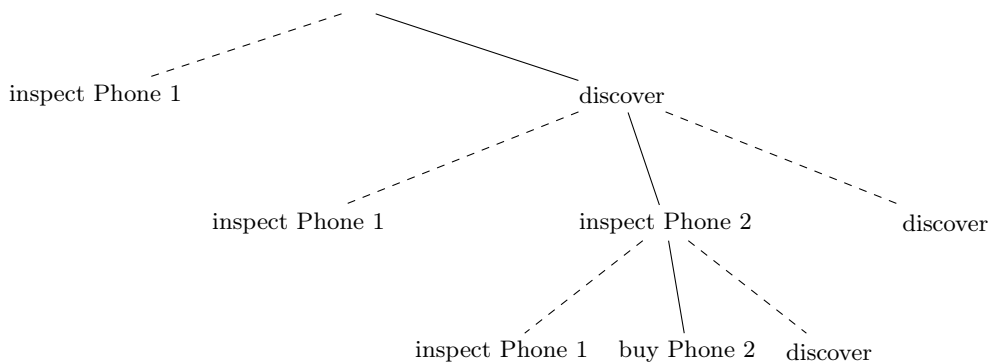

The ``search and discovery problem'' introduced in this paper formalizes
a consumer's dynamic decision process in this and similar settings.
The resulting framework allows to study settings that are difficult
to accommodate in existing search problems. In particular, neither
random \citep[e.g.][]{McCall1970} nor directed search \citep[e.g.][]{Weitzman1979} 
is well suited to study settings where
rational consumers remain oblivious to some, while obtaining only
partial information on other products. However, such settings are
common in practice. For example, online retailers and search intermediaries
present an abundance of alternatives on product lists that reveal partial information
only for some products. Consumers then decide between clicking on
products already discovered on the list to reveal full information, 
and browsing further to discover more products. More generally, in markets with a large number 
of alternatives, consumers  will remain unaware of many alternatives unless
they actively set out to discover more products. Similarly, in markets where 
rapid technological innovations
lead to a constant stream of newly available alternatives, few consumers
are aware of new releases without exerting effort to remain informed.

The contribution of this paper is to show that despite its complexity,
optimal search decisions and outcomes in the search and discovery
problem remain tractable if the consumer has stationary beliefs. 
First, I prove that the optimal policy is
fully characterized by reservation values similar to the well-known
reservation prices derived by \citet{Weitzman1979}. In each period,
a reservation value is assigned to each available action, and it is
optimal to always choose the action with the largest value. Each of
the reservation values is independent of any other available action
and can be calculated without having to consider expectations over
a myriad of future periods. Hence, reservation values remain tractable.
This allows to determine optimal search behavior under limited awareness
without using numerical methods.

Second, I prove that the purchase of a consumer solving the search
and discovery problem is equivalent to the same consumer having full
information and directly choosing products from a predetermined index.
This result generalizes the ``eventual purchase theorem'' derived
independently by \citet{Choi2016a}, \citet{Armstrong2017} and \citet{Kleinberg2016}
to the case of limited awareness.\footnote{\citet{Choi2016a} introduced the name 
and noted that ``Our eventual purchase theorem was
anticipated by \citet{Armstrong2015} and has been independently discovered
by \citet{Armstrong2017} and \citet{Kleinberg2016}.''} Similar to the eventual 
purchase theorem, my generalization allows
to derive a consumer's expected payoff and market demand without having
to consider a multitude of possible choice sequences that otherwise
make aggregation difficult. 

This paper also highlights several implications of limited awareness through 
a comparison of stopping decisions, expected payoffs and market demand with classical random 
and directed sequential search. 
A first implication of limited awareness is that it leads to 
two distinct search actions which posits a novel question: 
Do consumers benefit more from making it easier to discover more alternatives (e.g.
through search intermediaries), or from facilitating inspection by
more readily providing detailed product information? For the case where a consumer discovers 
one product at a time, I show that there exists a (possibly small) threshold for the number of 
alternatives after which  
the expected payoff increases more when facilitating discovery instead
of facilitating inspection. This highlights the relative importance of discovery costs 
in settings with many alternatives.


Moreover, limited awareness generates distinct patterns in the resulting 
market demand. In directed search, more consumers
preferring a product based on partial product information increases its market demand.
This need not be the case with limited awareness; if consumers remain unaware of a product, its market demand
does not increase as it becomes the preferred option. Whereas the same holds with random search, 
not being able to use partial information to decide whether to inspect a product 
induces consumers to stop earlier if total costs of revealing full product information remain the same. 

The search and discovery problem also provides an intuitive 
rationalization of ranking effects commonly observed in click-stream data \citep[e.g.][]{Ursu2018}: as consumers stop search 
before having discovered all products, products that would be discovered 
later are less likely to be bought. I show that  
these ranking effects are independent of the number of
available alternatives, and decrease as more products are discovered.
This mechanism offers a meaningful interpretation
of how advertising that provides partial product information 
is beneficial for a seller; \footnote{This relates to the ``informative view'' of advertising. See e.g.
\citet{Bagwell2007} for a summary and comparison to the ``persuasive
view''.} if a seller's marketing efforts make more consumers aware of a product
before search or increase the probability of the product being discovered
early on, ranking effects directly imply that they will increase the
demand. 

Finally, this paper adds to the empirical search literature by discussing 
implications of limited awareness for the estimation of structural search models. 
Besides highlighting differences in parameter estimates and counterfactual 
predictions across the three models, I show that a directed search model 
will lead to accentuated search cost estimates due to not accounting for limited awareness 
when rationalizing stopping decisions. 

The remainder of this paper is organized as follows. First, I discuss
related literature. Section 3 introduces the search and discovery
(henceforth SD) problem. Section 4 provides the optimal policy and discusses several 
extensions as well as limitations.
In Section 5, I generalize the eventual purchase theorem of \citet{Choi2016a}
and use this to derive a consumer's expected payoff as well as market
demand. Section 6 compares search problems and discusses empirical implications.
Section 7 concludes. Throughout, proofs are deferred to
the appendix.

\section{Related Literature}

The search and discovery problem introduced in this paper 
nests both classical random and directed sequential search as special cases. In random
search, a searcher has no prior information, searches randomly across
alternatives and decides when to end search \citep[e.g.][]{McCall1970,Lippman1976}.
In directed search, the searcher is aware of all available alternatives
and uses partial product information to determine an order in which
to inspect products and when to end search \citep[e.g.][]{Weitzman1979,Chade2006}.
In contrast, in the search and discovery problem, the consumer is
aware of only a few products. Hence, he not only decides in what order
to inspect products and when to end search, but also when to try to
discover more alternatives. 

To prove the optimality I use results from the multi-armed bandit literature
to first determine that a Gittins index policy is optimal,\footnote{\citet{Gittins2011} 
provide a textbook treatment of multi-armed bandit
problems and the Gittins index policy. As purchasing a product ends
search, search problems correspond to stoppable superprocesses as
introduced by \citet{Glazebrook1979}. } and then introduce a monotonicity condition to show that the Gittins
index reduces to simple reservation values. Specifically, I use the
results of \citet{Keller2003} who proved that a Gittins index policy
is optimal in their branching bandits framework. This framework differs
from the standard multi-armed bandit problem in that taking an action
will reveal information on multiple other actions. However, as an
action branches off into new actions and reveals information only
on those, the state of other available actions is never altered. Hence,
the important independence assumption continues to hold. 

Similar monotonicity conditions also apply in other multi-armed bandit
problems where they simplify the otherwise difficult calculation of
the Gittins index values \citep[see e.g. Section 2.11 in][]{Gittins2011}.
The present case differs in that monotonicity is only required for
the action of discovering more alternatives, but does not hold when
inspecting a product. In a recent working paper, \citet{Fershtman2019}
independently derived a similar characterization of the optimal policy
when applying a monotonicity condition in a general multi-armed bandit
problem where a decision maker also extends a set of alternatives. 

Moreover, monotonicity conditions also lead to the results in the
literature on (random) search problems where a searcher learns about
the distribution from which he is sampling \citep{Rothschild1974,Rosenfield1981,Bikhchandani1996}.
These authors determine priors and learning rules that satisfy a similar
condition based on which they can derive an optimal policy that is
myopic. The SD problem differs in that not all information about a
product is revealed when it is discovered such that it entails two distinct search actions. As I show, this makes it
difficult to find similar priors or learning rules that would lead
to a myopic optimal policy in extensions to the SD problem that incorporate
learning. 

Several other contributions extend Weitzman's (1979) seminal
search problem in different directions. \citet{Adam2001} studies
the case where the searcher updates beliefs about groups of alternatives
during search and finds a similar reservation value policy to be optimal.
\citet{Olszewski2015} generalize Pandora's rule to search problems
where the final payoff depends on all the alternatives that have been
inspected, not only the best one. Finally, \citet{Doval2018} analyzes
the optimal policy when a searcher can directly choose alternatives
without first inspecting them. 

This paper also relates to the recent literature studying problems
where a consumer gradually reveals more information on products \citep{Branco2010,Ke2016,Ke2019}.
These problems are formulated in continuous-time and generally do
not admit an optimal policy based on an index. The SD problem differs
in that it assumes that a consumer cannot purchase a product before
having revealed full information. This makes available actions independent
such that a tractable reservation value policy is optimal. Furthermore,
the SD problem allows that multiple products can be discovered at
a time such that with one action, information on multiple products
is revealed. Though \citet{Ke2016} also consider correlated payoffs,
discovering multiple products differs in that the correlation structure
of payoffs changes after the discovery; inspecting one does not reveal
information about other products discovered at the same time. 
 
The SD problem also subsumes decision processes considered in the
growing empirical literature estimating structural search models \citep[e.g.][]{Honka2014,Chen2016,Ursu2018}. Most
closely related are \citet{LosSantos2017} and \citet{Choi2019},
who also model consumers that decide between inspecting and revealing
more products. This paper differs in that I provide a tractable optimal
policy for the decision problem, whereas these studies use simplifying
assumptions and numerical methods to solve their models. The results
presented in this paper can serve as a justification for some of these
simplifying assumptions: Given that the optimal policy is myopic,
the one-step look-ahead approach adopted by \citet{LosSantos2017}
yields optimal choices of search actions if monotonicity holds. Moreover,
the optimal policy in the SD problem implies that as long as the consumer
has not yet revealed the last alternative, it will never be optimal
to go back and inspect a product that was discovered earlier if beliefs
are stationary. Hence,
the simplifying assumption made in \citet{Choi2019} where consumers
cannot go back and inspect a product revealed previously does not
affect the estimation as it would not be optimal to do so.

\citet{Honka2017a} and \citet{Morozov2019} also consider limited
awareness and assume that consumers cannot inspect products they are
not aware of. However, in their models consumers cannot discover products
beyond those they are initially aware of and the underlying search
problem then is equivalent to directed search. \citet{Koulayev2014}
estimates a search model where consumers also decide whether to reveal
more products, but assumes that revealing a product shows all information
on that product. Hence, there is no need for inspecting a product as
considered in this paper.\footnote{\citet{Koulayev2014} solves the dynamic decision problem using numerical
backwards induction. For the case where costs are increasing in time
(which is the case in his results), the present results suggest that
a simple index policy also characterizes the optimal policy for his
model.} 

Finally, related studies have highlighted other potential biases in 
search cost estimates. \citet{Jindal2020} show how heterogeneous prior beliefs 
can lead to an overestimation of search costs, \citet{Ursu2018} argues that 
that an incomplete search history also accentuates search cost estimates, whereas \citet{Yavorsky2020} discuss
the effects of normalizing search benefits.

\section[Model]{The Search and Discovery Problem \label{sec:Model}}

A risk-neutral consumer with unit demand faces a market offering a
(possibly infinite)\footnote{The problems with infinitely many arms in a multi-armed bandit problem
discussed by \citet{Banks1992} do not arise in the present setting.} number of products gathered in set $J$. Alternatives are heterogeneous
with respect to their characteristics. The consumer has preferences
over these characteristics which can be expressed in a utility ranking.
To simplify exposition and facilitate a comparison to existing models
from the consumer search literature \citep[e.g.][]{Armstrong2017,Choi2016a},
I assume that the consumer's ex post utility when purchasing alternative
$j$ is given by
\begin{equation}
u(x_{j},y_{j})=x_{j}+y_{j}
\end{equation}
where $x_{j}$ and $y_{j}$ are valuations derived from two distinct
sets of characteristics. Note, however, that the results presented
continue to hold for more general specifications that do not rely
on linear additive utility.\footnote{Specifically, suppose that when the consumer becomes aware of alternative
$j$, he reveals a signal on the distribution from which the utility
of $j$ will be drawn. Appropriately defining the distribution of
signals and the distribution of utilities conditional on these signals
then yields an equivalent search problem.} An outside option of aborting search without a purchase offering
$u_{0}$ is available.

The consumer has limited information on available alternatives. More
specifically, in periods $t=0,1,\dots$ the consumer knows both valuations
$x_{j}$ and $y_{j}$ only for products in a consideration set $C_{t}\subseteq J$.
For products in an awareness set\emph{ }$S_{t}\subseteq J$, the consumer
only knows partial valuations $x_{j}$. This captures the notion that
if the consumer is aware of a product, he has received some information
on the total valuation of the product. Finally, the consumer has no
information on any other product $j\in J\backslash\left(S_{t}\cup C_{t}\right)$. 

During search, the consumer gathers information by sequentially deciding
which action to take starting from period $t=0$. If the consumer decides
to discover more products, $n_{d}$ alternatives are added to the
awareness set. If less than $n_{d}$ alternatives have not yet been
revealed, only the remaining alternatives are revealed. For each of
the $n_{d}$ alternatives, the partial valuation $x_{j}$ is revealed.
To reveal the remaining characteristics of a product $j$, summarized
in $y_{j}$, the consumer has to inspect the product. This reveals
full information on the product and moves it from the awareness into
the consideration set. The latter implies $S_{t}\cap C_{t}=\emptyset$.

The order in which products are discovered is tracked by positions
$h_{j}\in\left\{ 0,1,\dots\right\} $, where a smaller position indicates
that a product is discovered earlier, and $h_{j}=0$ implies either
$j\in C_{0}$ or $j\in S_{0}$. Without loss of generality, it is
assumed that products are discovered in increasing order of their
index.\footnote{Note that in equilibrium settings, the order may be determined by
sellers' actions, requiring a careful analysis of how these will determine
the consumer's beliefs. For example, in online settings it is common
for sellers to bid on the position at which their product adverts
are shown \citep[see e.g.][]{Athey2011}. }

Two precedence constraints on the consumer's actions are imposed.
First, the consumer can only buy products from the consideration set.
Second, the consumer can only inspect products from the awareness
set. Whereas the first constraint is inherent in most search
problems and implies that a product cannot be bought before having obtained 
full information on it,\footnote{\citet{Doval2018} is a notable exception.} the
latter is novel to the proposed search problem. It implies that a
product cannot be inspected unless the consumer is aware of it. In
an online setting where a consumer browses through a list of products,
this constraint holds naturally: Individual product pages are reached
by clicking on the respective link on the list. Hence, unless a product
has been revealed on the list, it cannot be clicked on. In other environments,
this precedence constraint reflects that, unless a consumer knows
whether an alternative exists, he will not be able to direct search
efforts and inspect the specific alternative. For example, if a consumer
is not aware of a newly released phone model, he will not be able
to directly acquire detailed information before discovering it.

Given the setting and these constraints, the consumer decides sequentially
between the following actions: \vspace{1em}
\begin{enumerate}
\item \vspace{-1em}Purchasing any product from the consideration set $C_{t}$
and end search.
\item Inspecting any product from the awareness set $S_{t}$, thus revealing
$y_{j}$ for that product and adding it to the consideration set.
\item Discovering $n_{d}$ additional products, thus revealing their partial
valuations $x_{j}$ and adding them to the awareness set.
\end{enumerate}

The distinction between \emph{inspecting} and \emph{discovering} products is novel in the SD problem.
The two actions differ in three important ways. First, whereas the consumer can 
use product-specific information to decide 
the order in which to inspect products from the awareness set, 
the decision whether to discover more products is based solely on beliefs 
over products that may be discovered. 
Second, if $n_d > 1$, discovering products reveals information on 
multiple products. Finally, discovering products adds them into the awareness set,
whereas inspecting a product moves it into the consideration set. In combination with the
precedence constraints this implies that 
the actions that are available in the next period differ.\footnote{Note that the latter 
two points also imply that products that the consumer is not 
aware of cannot be modeled as a set of ex ante 
homogeneous products that differ in terms of beliefs and associated costs from the products in 
the awareness set.}

These actions are gathered in the \emph{set of available actions},
$A_{t}=C_{t}\cup S_{t}\cup \{d\}$, where $d$ indicates discovery. If
a consumer chooses an action $a=j\in C_{t}$, he buys product $j$,
whereas if he chooses an action $a=j\in S_{t}$, he inspects product
$j$. To clearly differentiate between the different types of actions,
this set can also be written as $A_{t}=\{b0,b3,s4,\dots,d\}$, where
$bj$ indicates purchasing and $sj$ inspecting product $j$.

Both inspecting a product and discovering more products is costly.
Inspection and discovery costs are denoted by $c_{s}>0$ and $c_{d}>0$
respectively. These costs can be interpreted as the cost of mental
effort necessary to evaluate the newly revealed information, or an
opportunity cost of the time spent evaluating the new information.
In line with this interpretation, I assume that there is free recall:
Purchasing any of the products from the consideration set does not
incur costs, and $c_{s}$ is the same for inspecting any of the products
in the awareness set. 

The consumer has beliefs over the products that he will discover,
as well as the valuation he will reveal when inspecting a product
$j$. In particular, $x_{j}$ and $y_{j}$ are independent (across
$j$) realizations from random variables $X$ and $Y$, where the
consumer has beliefs over their joint distribution. This implies that
the consumer believes that in expectation, products are equivalent.
A generalization where the distribution of $X$ depends on index $j$
is discussed in Section \ref{sec:Optimal-Search}. Note that throughout,
capital letters are used for random variables, lower case 
letters are used for the respective realizations and bold letters indicate vectors.

The consumer also has beliefs over the total number of available alternatives.
I assume that the consumer believes that with constant probability
$q\in[0,1]$, the next discovery will be the last.\footnote{Note that one can translate 
beliefs over a specific number of available alternatives to this probability by assuming 
it varies during search. For example, if $n_d=1$ and the consumer believes that there are 
3 alternatives in total, then $q_t=0$ when the consumer has not yet discovered the second alternative 
and $q_t=1$ otherwise. A specification like this (and any specification where 
$q_t\leq q_{t+1} \forall t$) also satisfies the monotonicity 
condition \eqref{eq:condition_mono} presented in Appendix \ref{sec:Generalizations}. 
Consequently, if it is assumed that the consumer knows $\left|J\right|$, monotonicity continues to hold. } 
As shown in the next section, the optimal policy is independent of
the number of remaining discoveries that may be available in the future.
Note, however, that this belief specification implicitly assumes that
the consumer always knows whether he can reveal $n_{d}$ more alternatives.
An extension presented in Section \ref{sec:Optimal-Search} covers
the case where the consumer does not know how many alternatives will
be revealed.

All information the consumer has in period $t$ is summarized in the
information tuple $\Omega_{t}=\left\langle \bar{\Omega},\omega_{t}\right\rangle $.
The tuple $\bar{\Omega}=\left\langle u(x,y),n_{d},c_{d},c_{s},G_{X}(x),F_{Y|X=x}(y),q\right\rangle $
represents the consumer's knowledge and beliefs on the setting. It contains the utility
function, how many products are discovered, and the different costs.
It also contains the consumer's beliefs summarized in the probability $q$ and the
cumulative densities $G_{X}(x)$ and $F_{Y|X=x}(y)$. The latter specifies
the cumulative density of $Y$, conditional on the realization of
$X$, which is observed by the consumer before choosing to inspect
a product. As a short-hand notation, I use $G(x)$ and $F(y)$ for
these distributions. As a regularity condition, it is assumed that
both $G(x)$ and $F(y)\forall x$ have finite mean and variance. 

During search, the consumer reveals valuations $x_{j}$ and $y_{j}$
for the various products. This information is tracked in the set $\omega_{t}$,
containing realizations $x_{j}$ for $j\in S_{t}\cup C_{t}$ and $y_{j}$
for $j\in C_{t}$. The set of available actions $A_{t}$ and the information
tuple $\Omega_{t}$ capture the state in $t$. The consumer's initial information 
on the alternatives are captured in $\omega_{0}$ which will contain (partial) valuations 
of products in the initial awareness and consideration set. 
Figure \ref{fig:Transition-of-state}
shows their transitions starting from period $t=0$. The depicted
example assumes that there are only two alternatives available and
that products are discovered one at a time. If the consumer initially
chooses the outside option ($b0$), no new information is revealed,
and no further actions remain. If the consumer instead reveals the
first alternative, he can inspect it in $t=1$. 

\begin{figure}[tbh]
\centering{}%
\begin{minipage}[t]{0.8\columnwidth}%
\hspace{-1.4cm}
\scalebox{0.93}{
\begin{tikzpicture}[grow=right,every tree node/.style={inner sep=0pt,font=\normalsize,align=left},
	   level distance=16.3em,sibling distance=4em,baseline, label = a,
   edge from parent path={(\tikzparentnode) -- (\tikzchildnode)}]

	\Tree [.\node[]{$\Omega_0=\left\langle \bar\Omega,\left\{ x_0,y_0 \right\} \right\rangle$  \\$A_0 = \{b0,d\}$};
				\edge node[auto=right]{$d$};
				[.\node[]{$\Omega_1=\left\langle \bar\Omega,\left\{ x_0,y_0,x_1 \right\} \right\rangle$ \\$A_1 = \{b0,s1,d\}$};
					\edge[] node[auto=right]{$d$};
					\node[]{$\Omega_2 = \left\langle \bar\Omega,\left\{ x_0,y_0,x_1,x_2 \right\} \right\rangle$\\$A_2 = \{b0,s1,s2\}$};
					\edge node[auto=left]{$s1$};
					[.\node[]{$\Omega_2 = \left\langle \bar\Omega,\left\{ 					x_0,y_0,x_1,y_1 \right\} \right\rangle$ \\ $A_2 = \{b0,b1\}$};]
					\edge[] node[auto=left]{$b0$};
					\node[]{$\Omega_2=\left\langle \bar\Omega,\left\{ x_0,y_0,x_1 \right\} \right\rangle$\\$A_2 = \emptyset$};]
				\edge[] node[auto=left]{$b0$}; 
				\node[]{$\Omega_1=\left\langle \bar\Omega,\left\{ x_0,y_0 \right\} \right\rangle$\\$A_1 = \emptyset$};] 
\end{tikzpicture}
}

\caption{Transition of state variables $\Omega_{t}$ (information tuple) and
$A_{t}$ (set of available actions) for $n_{d}=1$ and $\left|J\right|=2$.
\label{fig:Transition-of-state}}
\end{minipage}
\end{figure}
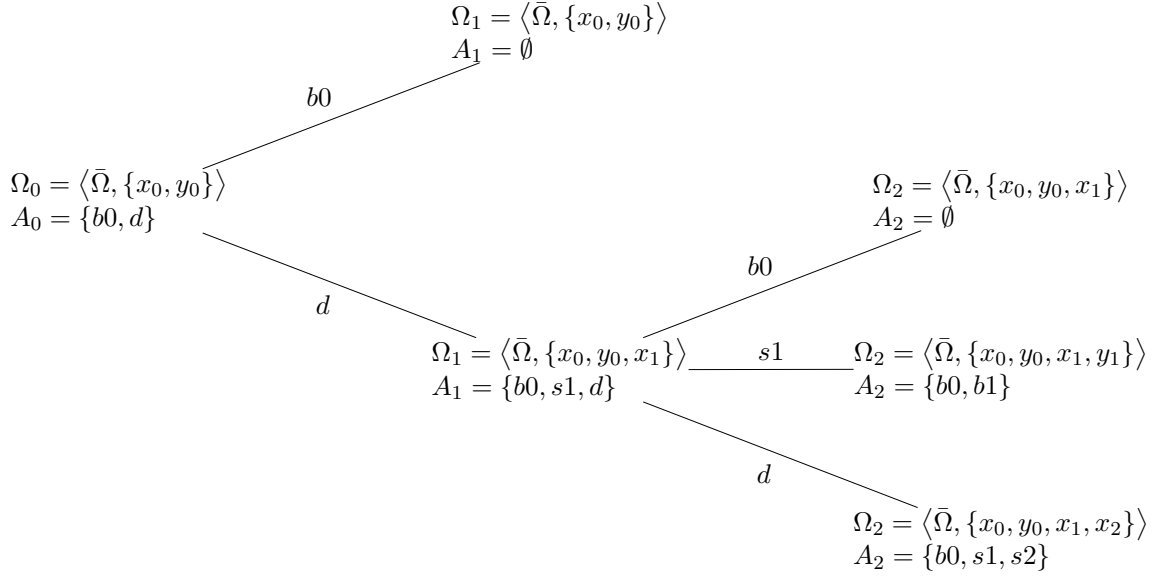

\subsection{The Consumer's Dynamic Decision Problem}

The setting above describes a dynamic Markov decision process, where
the consumer's choice of action determines the immediate rewards,
as well as the state transitions. The state in $t$ is given by $\Omega_{t}$
and $A_{t}$. As the valuations $x_{j}$ and $y_{j}$ can take on
any (finite) real values, the state space in general is infinite.\footnote{An exception is when $x_{j}$ and $y_{j}$ are drawn from discrete
distributions, which limits the number of possible valuations that
can be observed.} Time $t$ itself is not included in the state; 
given $A_t$ and $\Omega_{t}$, it is irrelevant to the agent's choice, 
because beliefs (over valuations and termination of discovery) are time invariant. 

The consumer's problem consists of finding a feasible sequential policy,
which maximizes the expected payoff of the whole decision process.
A feasible sequential policy selects an action $a_{t} \in A_{t}$ given 
information in $\Omega_{t}$ in each period $t$. Let $\Pi$ denote the set containing all feasible policies.
Formally, the consumer solves the following dynamic programming problem
\begin{equation}
\max_{\pi\in\Pi}V(\Omega_{0},A_{0};\pi)
\end{equation}
where $V(\Omega_{t},A_{t};\pi)$ is the value function defined as the 
expected total payoff of following policy $\pi$ starting from the state in $t$. Let 
\begin{equation}
\left[B_{a}V\right](\Omega_{t},A_{t};\pi)=R(a)+\mathbb{E}_{t}\left[V(\Omega_{t+1},A_{t+1};\pi)|a\right]
\end{equation}
denote the Bellman operator, where the immediate rewards $R(a)$ either
are inspection costs, discovery costs, or the total valuation of a
product $j$ if it is bought. Immediate rewards $R(a)$ therefore
are known for all available actions. $\mathbb{E}_{t}\left[V(\Omega_{t+1},A_{t+1};\pi)|a\right]$
denotes the expected total payoff over the whole future, conditional
on policy $\pi$ and having chosen action $a$.\footnote{In this formulation of the problem, the consumer does not discount
future payoffs. This is in line with the consumer search literature,
which usually assumes a finite number of alternatives without discounting.
However, it is straightforward to show that the results continue to
hold if a discount factor $\beta<1$ is introduced. In this case,
the search and discovery values defined in the next section need to
be adjusted accordingly. } The expectations operator integrates over the respective distributions
of $X$ and $Y$. A purchase in $t$ ends search such that $A_{t+1}=\emptyset$
and $\mathbb{E}_{t}\left[V(\Omega_{t+1},\emptyset;\pi)|a\right]=0$
whenever $a\in C_{t}$. The corresponding Bellman equation is given
by

\begin{equation}
V(\Omega_{t},A_{t};\pi)=\max_{a\in A_{t}}\left[B_{a}V\right](\Omega_{t},A_{t};\pi)
\end{equation}

\section{Optimal Policy \label{sec:Optimal-Search}}

The optimal policy for the SD problem is fully characterized by three
reservation values. In what follows, I first define these reservation
values, before stating the main result. At the end of this section,
I discuss possible extensions based on a monotonicity condition, as
well as limitations. 

As in \citet{Weitzman1979}, suppose there is a \emph{hypothetical
}outside option offering utility $z$. Furthermore, suppose the consumer
faces the following comparison of actions: Immediately take the outside
option, or inspect a product with known $x_{j}$ and end search thereafter.
In this decision, the consumer will choose to inspect alternative
$j$ whenever the following holds: 
\begin{equation}
Q_{s}(x_{j},c_{s},z)\equiv\mathbb{E}_{Y}\left[\max\{0,x_{j}+Y-z\}\right]-c_{s}\geq0
\end{equation}
\emph{$Q_{s}(x_{j},c_{s},z)$ }defines the expected \emph{myopic}
net gain of inspecting product $j$ over immediately taking the outside
option. If the realization of $Y$ is such that $x_{j}+y_{j}\leq z$,
the consumer takes the hypothetical outside option after inspecting
$j$ and the gain is zero. When $x_{j}+y_{j}>z$, the gain over immediately
taking the hypothetical outside option is $x_{j}+y_{j}-z$. The expectation
operator $\mathbb{E}_{Y}\left[\cdot\right]$ integrates over these
realizations.

The \emph{search value} of product $j$, denoted by $z_{j}^{s}$,
then is defined as the value offered by a hypothetical outside option
that makes the consumer indifferent in the above decision problem.
Formally, $z_{j}^{s}$ satisfies 
\begin{equation}
Q_{s}(x_{j},c_{s},z_{j}^{s})=0\label{eq:res_search}
\end{equation}
which has a unique solution \citep[see Lemma 1 in][]{Adam2001}. The
search value can be calculated as 
\begin{equation}
z_{j}^{s}=x_{j}+\xi\label{eq:res_search_explicit}
\end{equation}
where $\xi$ solves $\int_{\xi}^{\infty}\left[1-F(y)\right]\mathrm{d}y-c_{s}=0$
(see Appendix \ref{sec:Details-on-the}).

The \emph{purchase value} of product\emph{ $j$}, denoted by $z_{j}^{b}$,
is defined as the utility obtained when buying product $j$: 
\begin{equation}
z_{j}^{b}=u(x_{j},y_{j})\label{eq:res_buy}
\end{equation}

Based on reservation values given by \eqref{eq:res_search} and \eqref{eq:res_buy},
\citet{Weitzman1979} showed that it cannot be optimal to inspect
a product that does not offer the largest search value, or to stop
when the largest remaining search value exceeds the largest purchase
value. Hence, for given $S_{t}$ and $C_{t}$, it is optimal to always
inspect and buy in decreasing order of search and purchase values.
However, this rule does not fully characterize an optimal policy in
the SD problem, as the consumer can additionally discover more alternatives.

For this additional action, a third reservation value based on a similar
myopic comparison is introduced. Suppose the consumer faces the following
comparison of actions: Take a hypothetical outside option offering
$z$ immediately, or discover more products and then search among
the newly revealed products. The consumer will choose the latter whenever
the following holds: 
\begin{equation}
Q_{d}(c_{d},c_{s},z)\equiv\mathbb{E}_{\boldsymbol{X}}\left[V\left(\left\langle \bar{\Omega},\omega(\boldsymbol{X},z)\right\rangle ,\left\{ b0,s1,\dots,sn_{d}\right\} ;\tilde{\pi})\right)\right]-z-c_{d}\geq0\label{eq:exp_net_ben_disc}
\end{equation}
where $\omega(\boldsymbol{X},z)=\left\{ z,x_{1},\dots,x_{n_{d}}\right\} $
denotes the information the consumer has after revealing the $n_{d}$
more alternatives and $\tilde{\pi}$ is the policy that optimally inspects the $n_d$ discovered products.
Note that with some abuse of notation, product
indices were adjusted to the reduced decision problem, such that $j=0,1,\dots,n_{d}$
indicates the hypothetical outside option and the newly revealed products.
 
$Q_{d}(c_{d},c_{s},z)$ defines the \emph{myopic} net gain of discovering
more products and optimally searching among them over immediately
taking the outside option. It is myopic in the sense that it ignores
the option to continue searching beyond the products that are discovered.
In particular, note that $V\left(\left\langle \bar{\Omega},\omega(\boldsymbol{X},z)\right\rangle ,\left\{ b0,s1,\dots,sn_{d}\right\} ;\tilde{\pi})\right)$
is the value function of having an outside option offering $z$ and
optimally inspecting alternatives for which partial valuations in
$\boldsymbol{X}$ are known. Possible future discoveries and any products
in $S_{t}$ or $C_{t}$ are excluded from the set of available actions
in this value function. This implies that the discovery value does not depend 
on the consumer's beliefs over whether the next discovery will be the last.
Finally, $\mathbb{E}_{\boldsymbol{X}}\left[\cdot\right]$
defines the expectation operator integrating over the joint distribution
of the partial valuations in $\boldsymbol{X}$. Formal details on
the calculation of the expectations and the value function are provided
in Appendix \ref{sec:Details-on-the}.

As for the search value, let the \emph{discovery value}, denoted by
$z^{d}$, be defined as the value of the hypothetical outside option
that makes the consumer indifferent in the above decision. Formally,
$z^{d}$ is such that 
\begin{equation}
Q_{d}(c_{d},c_{s},z^{d})=0\label{eq:res_exp}
\end{equation}
which has a unique solution. In the case where $Y$ is independent
of $X$, the discovery value can be calculated as 
\begin{equation}
z^{d}=\mu_{X}+\Xi(c_{s},c_{d})
\end{equation}
 where $\mu_{X}$ denotes the mean of $X$ and $\Xi(c_{s},c_{d})$
solves \eqref{eq:res_exp} for an alternative random variable $\tilde{X}=X-\mu_{X}$.
Further details for the calculation are provided in Appendix \ref{sec:Details-on-the}.

Theorem \ref{theo:optimal_policy} provides the first main result.
It states that the optimal policy for the search problem reduces to
three simple rules based on a comparison of the search, purchase and
discovery values. In particular, the rules imply that in each period
$t$, it is optimal to take the action with the largest reservation
value defined in \eqref{eq:res_search}, \eqref{eq:res_buy}, and
\eqref{eq:res_exp}. Hence, despite being fully characterized by myopic
comparisons to a hypothetical outside option, these reservation values
rank the expected payoffs of actions over all future periods. 
\begin{theorem}
\label{theo:optimal_policy} Let $\tilde{z}^{b}(t)=\max_{k\in C_{t}}u(x_{k},y_{k})$
and $\tilde{z}^{s}(t)=\max_{k\in S_{t}}z_{k}^{s}$ denote the largest search and purchase 
values in period $t$. An
optimal policy for the search and discovery problem is characterized
by the following three rules:
\begin{itemize}
\item \textsc{Stopping rule}: Purchase $j\in C_{t}$ and end search whenever
$z_{j}^{b}=\tilde{z}^{b}(t)\geq\max\left\{ \tilde{z}^{s}(t),z^{d}\right\} .$
\item \textsc{Inspection rule}: Inspect $j\in S_{t}$ whenever $z_{j}^{s}=\tilde{z}^{s}(t)\geq\max\left\{ \tilde{z}^{b}(t),z^{d}\right\} $.
\item \textsc{Discovery rule}: Discover more products whenever $z^{d}\geq\max\left\{ \tilde{z}^{b}(t),\tilde{z}^{s}(t)\right\} $.
\end{itemize}
\end{theorem}
The proof of Theorem \ref{theo:optimal_policy} relies on results
from the literature on multi-armed bandit problems, specifically the
branching bandits framework of \citet{Keller2003}. These authors
show that in a multi-armed bandit problem where taking an action branches
off into new actions, a Gittins index policy is optimal. Importantly,
as an action branches off, it cannot be taken again in its original
state. This ensures that available actions are independent in the
sense that taking one does not alter the state of any other available
action. The imposed precedence constraints combined with the fact
that the consumer cannot discover a product for a second time imply
the same branching structure in the SD problem, and the results of
\citet{Keller2003} therefore imply that a Gittins index policy is
optimal. Introducing a monotonicity condition I then show that the
Gittins index is equivalent to the simple reservation values defined
above. 

Based on Theorem \ref{theo:optimal_policy}, optimal search behavior
can be analyzed using only \eqref{eq:res_search}, \eqref{eq:res_buy}
and \eqref{eq:res_exp}. \citet{Weitzman1979} showed that search
values decrease in inspection costs and increase if larger realizations
$y_{j}$ become more likely through a shift in the probability mass
of $Y$. The same applies to the discovery value. It decreases in
discovery costs and increases if probability mass of $X$ is shifted
towards larger values. The discovery value also depends on inspection
costs and the conditional distribution of $Y$ through the value function;
it decreases in inspection costs and increases if larger values of
$Y$ are more likely. 

To see the latter, consider the case where alternatives are discovered
one at a time. In this case, the myopic net gain of discovering more
products reduces to
\begin{equation}
Q_{d}(c_{d},c_{s},z)=\mathbb{E}_{X}\left[\max\left\{ 0,Q_{s}(X,c_{s},z)\right\} \right]-c_{d}
\end{equation}
For any $c_{s}^{\prime}>c_{s}$, it holds that $Q_{s}(x,c_{s}^{\prime},z)\leq Q_{s}(x,c_{s},z)$
for all finite values of $x$ and $z$, implying that $Q_{d}(c_{d},c_{s}^{\prime},z)\leq Q_{d}(c_{d},c_{s},z)$
for all $z$. As $Q_{d}(c_{d},c_{s},z)$ is decreasing in $z$ (see
Appendix \ref{sec:Proofs}), it follows that the
respective discovery values satisfy $z^{d\prime}\leq z^{d}$.

The optimal policy being fully characterized by simple rules leads
to straightforward analysis of optimal choices for any given awareness
and consideration sets. For example, consider a period $t$ where
$\max\left\{ z^{d},\tilde{z}^{s}(t)\right\} <\tilde{z}^{b}(t)$ such
that the consumer stops searching. When decreasing inspection costs
sufficiently in this case, the inequality reverts and the consumer
will instead either first discover more products, or inspect the best
product from the awareness set.

\subsection{Monotonicity and Extensions}
 
For the reservation value policy of Theorem \ref{theo:optimal_policy}
to be optimal, the discovery value needs to fully capture the expected
net benefits of discovering more products, including the option value
of being able to continue discovering products. The monotonicity condition
used in the proof of the theorem ensures that this holds. It states
that the expected net benefits of discovering more products do not
increase during search. Hence, whenever the consumer is indifferent
between taking the hypothetical outside option and discovering more
products in $t$, he will either continue to be indifferent or take
the outside option in $t+1$. Whether the consumer can continue to
discover products in $t+1$ thus does not affect expected net benefits
in $t$, and the discovery value fully captures the expected net benefits.\footnote{For the search and purchase values, no monotonicity condition is required.
This follows from the fact that in the independent comparison to the
hypothetical outside option, both actions do not provide the option
to continue searching. After buying a product, search ends, and after
having inspected a product, the only option that remains is to either
buy the product or choose the hypothetical outside option. Consequently,
for inspection and purchase, at most one future period needs to be
considered to fully capture the respective net benefits over immediately
taking the outside option.}

In the baseline SD problem, several assumptions directly imply that
the monotonicity condition holds. Specifically, (i) the consumer believes
that product valuations are independent and identically distributed,
(ii) $q$ remains constant and (iii) $n_{d}$ is known. However, these
assumptions can be relaxed to capture a wider range of settings. Below,
three related extensions are presented. Formal results and further
details are presented in Appendix \ref{sec:Generalizations}.

\textbf{Ranking in distribution:} In some settings, the consumer'
beliefs are such that the distribution of partial valuations depends
on the position at which a product is discovered. Monotonicity will
be satisfied if beliefs are such that the mean of $X_{j}$ decreases
in a product's position $h_{j}$, or more generally if beliefs are
such that $X_{j}$ first-order stochastically dominates $X_{k}$ if
$h_{j}\leq h_{k}$. The optimal policy then continues to be characterized
by Theorem \ref{theo:optimal_policy}, the only difference being that
the discovery value is based on the position-specific beliefs and
decreases during search, making it optimal to recall products in some
cases. This could result in a market environment where sellers
of differentiated products compete in marketing efforts for consumers
to become aware of their products early on. If sellers offering better
valuations have a stronger incentive to be discovered first, they
will increase marketing efforts.\footnote{See, for example, the discussion on 
non-price advertising and the
related references cited in \citet{Armstrong2017}. } 
Consumers' beliefs then will
reflect this ordering such that monotonicity holds and the simple
optimal policy can be used to characterize equilibria. Similarly, online stores 
often use algorithms to first present products
that consumers may like more. This again satisfies monotonicity such that the
tractable optimal policy can be used to rationalize search behavior in click-stream data from 
such stores. 

\textbf{Unknown }$\boldsymbol{n_{d}}$\textbf{:} In other environments,
a consumer may not know how many alternatives he will discover. For
example, a consumer may believe that there are still alternatives
he is not aware of and thus try to discover them, only to realize
that he already is aware of all the available alternatives. In such
cases, a belief over how many alternatives are going to be discovered
needs to be specified. The reservation value policy continues to be
optimal if these beliefs are such that monotonicity is satisfied.
This will be the case if beliefs are constant, or if (more realistically)
the consumer expects to discover fewer alternatives the more alternatives
he already has discovered.\footnote{This would reflect the case where the consumer expects it to become
harder to discover alternatives the fewer alternatives have not yet
been discovered. Alternatively, this could be modeled as either $q$
or $c_{d}$ to increase with each discovery, which also satisfies
monotonicity. } The only difference to the baseline is that in $Q_{d}(c_{d},c_{s},z)$,
expectations are additionally based on beliefs over how many alternatives
will be revealed.

\textbf{Multiple discovery technologies:} Consumers may also have
multiple discovery technologies at their disposal. In an online setting,
for example, each technology may represent a different online shop
offering alternatives. Moreover, advertising measures may separate
products into different product pools. In such settings, the consumer
also decides which technology to use to discover more alternatives.
By assigning each of the discovery technologies a different discovery
value, the optimal policy can be adjusted to accommodate this case.\footnote{An interesting extension for future research is to model the case
where a consumer can choose the order in which products are revealed
based on a product characteristic such as price. This requires modeling
beliefs that reflect this ordering through updating the support of
the price distribution; in an ascending order the minimum price that
can be discovered needs to increase with every discovery. \citet{Chen2016}
incorporate choices of such search refinements in their empirical
model. However, in their model, a consumer simultaneously decides
on the refinement and which position to inspect. In contrast, if such
choices are modeled as a SD problem the consumer would sequentially
decided between a discovery technology and whether to inspect a product.
This is more closely done by \citet{LosSantos2017}, who also model
sequential choice of search refinements and clicks, but use simplifying
assumptions and do not derive the optimal policy. }

\subsection{Limitations\label{subsec:Limitations}}

Though the optimal policy applies to a broad class of search problems,
two limitations exist. The first is that in the dynamic decision process,
all available actions need to be independent of each other; performing
one action in $t$ should not affect the payoff of any other action
that is available in $t$. This is required to guarantee that the
reservation values fully capture the effects of each action. Recall
that each reservation value does not depend on the availability of
other actions. If independence does not hold, however, the availability
of other actions also influences the expected payoff of an action.
Choosing actions based only on reservation values that disregard these
effects therefore will not be optimal. Alternative search problems that violate this independence
assumption are presented in the appendix.

The second limitation is that the monotonicity condition discussed
above needs to hold for the discovery value to be based on \emph{myopic}
net benefits. If this condition does not hold, then the discovery
value does not fully capture the expected net benefits of discovering
more products. However, as long as independence of the available actions
is satisfied, a Gittins index policy remains optimal (see proof of
Theorem 1). Hence, the optimal policy when monotonicity fails consists
of comparing the search and purchase values from equations \eqref{eq:res_search_explicit}
and \eqref{eq:res_buy} with the Gittins index value for discovery
that explicitly accounts for future discoveries. 

One interesting case where this fails is if the consumer learns about
the distribution of $X$ or the number of alternatives he will discover
during search. So far, it was assumed that independent of the information
the consumer reveals during search, his beliefs remain unchanged.
This will be the case if either the consumer has rational expectations
and hence knows the underlying distributions, or simply does not update
beliefs. With learning, the consumer updates his beliefs based on
partial valuations or number of products revealed in a discovery. 

Similar learning models have been studied in the context of classic
search (and stopping) problems, where the consumer learns about the
distribution he is sampling from \citep[e.g.][]{Rothschild1974,Rosenfield1981,Bikhchandani1996,Adam2001}.
\footnote{The SD problem is equivalent to these learning problems in the case
where $c_{s}=0$ and the consumer updates beliefs about the distribution
of the random variable $X+Y$.  } Whereas these studies determine prior beliefs or learning rules such
that the optimal policy is based on myopic reservation values, similar
conditions do not guarantee that monotonicity holds in a SD problem
where a consumer learns about the distribution of $X$ or the number
of alternatives he will discover. The reason is that in classic search
problems a consumer reveals full information when inspecting a product.
Hence, if a product turns out to be a good match, the value of stopping
increases along with the value of continuing search, where the learning
rule guarantees that this is such that the expected net benefits of
continuing search over stopping with the current best option weakly
decrease with each inspection.\footnote{See e.g. Theorem 1 in \citet{Rosenfield1981}. }
In contrast, in the SD problem, discovering either more or better
partial valuations does not necessarily increase the value of the
best option in the consideration set.\footnote{If the consumer 
learns about the distribution of $Y$ conditional on $X$, then 
discovering more alternatives with similar $X$ can increase the value of the best
option. Analyzing this mechanism provides an interesting avenue for future research. } 
For example, the consumer can
discover many products that look very promising based on partial valuations,
but after inspection realize that these products are a bad match after
all. In this case, the value of stopping remains the same, whereas
beliefs are shifted such that the consumer expects to find better
or more products in future discoveries.

Extending the SD problem to the case where the consumer learns about
the distribution of $X$ or the number of alternatives therefore comes
at the cost of losing tractability of the discovery value; a tractable
expression for the Gittins index value for the discovery action (henceforth
denoted by $z_{t}^{L}$) is difficult to obtain as it is necessary
to determine the value function of a dynamic decision process that
includes many future periods. Moreover, whereas the discovery value
in Theorem \ref{theo:optimal_policy} remains constant throughout
search, $z_{t}^{L}$ changes whenever the consumer updates beliefs.
Consequently, the optimal policy when the consumer updates beliefs
becomes more complex in that the discovery value changes with each
discovery and explicitly includes future periods. 

Whereas $z_{t}^{L}$ is not tractable and computationally expensive
to obtain, it is possible to derive bounds on this value that are
easier to compute and can serve as an approximation. First, $z_{t}^{L}$
can be approximated from below through \emph{k}-step look-ahead values.
The 1-step look-ahead value is defined by \eqref{eq:res_exp}, where
the expectation operator is adjusted to account for the consumer's
beliefs in $t$. As \emph{k }increases, more future discoveries are
considered in \eqref{eq:res_exp}, leading to a more precise approximation
of $z_{t}^{L}$ up to the point where $z_{t}^{L}$ is calculated precisely.
Second, a result of \citet{State1970} can be used to derive an upper
bound. These authors show that the expected value of continuing search
when the consumer fully resolves uncertainty on the underlying distributions
in the next period exceeds the true continuation value in a classic
search problem where a consumer samples from an unknown distribution.
The same logic directly applies in the extension to the SD problem
and the upper bound then can be computed using the results provided
in the next section. A formal treatment of these bounds is provided
in the appendix.

\section{Eventual Purchases, Consumer's Payoff, and Demand}

In an environment where consumers sequentially inspect products, a
consumer's expected payoff and the market demand results from integrating
over different possible choice sequences leading to eventual purchases.
Conceptually, this poses a major challenge, as the number of possible
choice sequences grows extremely fast in the number of available alternatives.\footnote{For example, with only one alternative and an outside option, there
are four possible choice sequences. With two alternatives, the number
of possible choice sequences increases to 20, and with three alternatives,
there are already more than 100 possible choice sequences. }

Theorem \ref{theo:effective_values} allows to circumvent this difficulty.
It states that the purchase outcome of a consumer solving the search
problem is equivalent to a consumer directly buying a product that
offers the highest \emph{effective value}. Importantly, a product's
effective value does not depend on the various possible choice sequences
leading to its purchase.
\begin{theorem}
Let 
\[
w_{j}\equiv\begin{cases}
u_{j} & \text{if } u_j < z^d \text{ and } j\in C_{0}\\
\tilde{w}_{j} & \text{if }\tilde{w}_{j}<z^{d}\text{ or }j\in S_{0}\\
z^{d}+f(h_{j})+\varepsilon\tilde{w}_{j} & \text{else}
\end{cases}
\]
be the effective value for product $j$ revealed on position $h_{j}$
where $\tilde{w}_{j}\equiv\min\{z_{j}^{s},z_{j}^{b}\}=x_{j}+\min\left\{ \xi,y_{j}\right\} $,
$f(h_{j})$ is a non-negative function and strictly decreasing in
$h_{j}$ and $\varepsilon$ is an infinitesimal. The solution to the
search and discovery problem with initial consideration set $C_{0}$
and awareness set $S_{0}$ leads to the eventual purchase of the product
with the largest effective value. \label{theo:effective_values}
\end{theorem}
This result is based on and generalizes the ``eventual purchase theorem''\emph{
}of \citet{Choi2016a} \citep[and independently][]{Armstrong2017,Kleinberg2016} 
to the case where the consumer has limited
awareness. The value $\tilde{w}_{j}$ used in the theorem is equivalent
to the effective value defined by \citet{Choi2016a}, and the proof
follows the same logic; as a product (incl. out outside option) is
always bought, the proof only needs to establish that the optimal
policy never prescribes to buy a product that does not have the largest
effective value. 

The generalization to the case of limited awareness follows from the
following implication of the optimal policy: Whenever both the inspection
and the purchase value of a product in the awareness set exceed the
discovery value, the consumer will buy the product and end search.
Hence, when $\tilde{w}_{j}\geq z^{d}$, the consumer never discovers
products on positions beyond $h_{j}$. This is captured in the effective
values by the term $z^{d}+f(h_{j})$, which ranks alternatives based
on when during search they are discovered, yielding a larger effective
value if a product is discovered earlier. The infinitesimal in the
last condition additionally is necessary to rank products that are
revealed on the same position. Suppose we have $\tilde{w}_{j}>\tilde{w}_{k}\geq z^{d}$
for two products discovered on the same position. Without the infinitesimal,
the effective value would be $w_{j}=w_{k}$, implying the consumer
would be indifferent between buying either of the two products. This
contrasts the optimal policy, which for $\tilde{w}_{j}>\tilde{w}_{k}$ will
never prescribe to buy $k$ if both $j$ and $k$ are in the awareness
set. If $n_{d}=1$, the infinitesimal is not required. 

The result continues to hold for extensions of the SD problem, as
long as the discovery values are predetermined. The only difference
then is that in the effective value of an alternative $j$, the discovery
value depends on the position at which $j$ is revealed.

\subsection{Expected Payoff}

Based on these results, it is now possible to derive a simple characterization
of a consumer's expected payoff, as summarized in Proposition \ref{prop:Expected-consumer-welfare}.
In this expression, the expected payoff does not explicitly depend
on inspection and discovery costs; they affect the expected payoff
only through the discovery and search values. As the proof shows,
this follows from the definition of these values, which relate expected
payoffs and costs \citep[as in][]{Choi2016a}. Based on this characterization,
it is only necessary to derive the distribution of the effective values
without having to explicitly consider different choice sequences.
Note also that as the effective value is adjusted, the expected payoff
does not depend on the choice of function $f(h)$ which ranks alternatives
based on their position in the effective value. 
\begin{proposition}
\label{prop:Expected-consumer-welfare}A consumer's expected payoff
in the SD problem is given by 
\[
V(\Omega_{0},A_{0};\pi)=\mathbb{E}_{\hat{\boldsymbol{W}}}\left[\max_{j\in J}\hat{W}_{j}\right]
\]
where $\mathbb{E}_{\hat{\boldsymbol{W}}}\left[\cdot\right]$ integrates
over the distribution of $\hat{\boldsymbol{W}} = \left[\hat{W}_{0},\dots,\hat{W}_{\left|J\right|}\right]^{\prime}$, with $\hat{w}_{j}$
being the effective value adjusted with $\hat{w}_j = u_j \forall j \in C_0$,$\hat{w}_j = \tilde{w}_{j} \forall j \in S_0$, and $f(h_{j})=\varepsilon=0\forall h_{j}$.
If $\left|J\right|=\infty$, $V(\Omega_{0},A_{0};\pi)=z^{d}$. 
\end{proposition}
Whereas it is clear that making either inspection or discovery easier
leads to an increase in the expected payoff, it is not obvious which
of these two changes is more beneficial for a consumer. For the case
where $n_{d}=1$, Proposition \ref{prop:welfare-change} shows that
if the number of alternatives exceeds some threshold, then the consumer
benefits more from facilitating the discovery of additional products.\footnote{Note that this threshold can be zero. For example, this is the case
when $u_{0}=0$, $c_{s}=0.1$ and $c_{d}=0.1$, and the valuations
are drawn from standard normal distributions.}
\begin{proposition}
\label{prop:welfare-change} If $n_{d}=1$, there exists a threshold
$n^{*}$ such that whenever $\left|J\right|>n^{*}$, a consumer benefits
more from a decrease in discovery costs than a decrease in inspection
costs. This threshold decreases in the value of the alternatives in
the initial consideration and awareness set.
\end{proposition}
Whereas the proof is more involved, the intuition is that when there
are only few alternatives available, the consumer is more likely to
first discover all alternatives and then start inspecting alternatives.
Hence in expectation, he pays the inspection costs relatively often
and a reduction in inspection costs will be more beneficial. Similarly,
when the value of the outside option is large, the consumer is likely
to inspect fewer of the products he discovers, leading to relatively
small benefits of a reduction in inspection costs. 

For settings where $n_{d}>1$, it becomes difficult to obtain similarly
general results. In particular, for some distributions and $n_{d}$,
it is possible that decreasing inspection costs increases the discovery
value $z^{d}$ by more than decreasing the discovery costs by the
same amount. In such cases, the consumer will benefit more from making
inspection less costly. Nonetheless, the general intuition remains
the same in such settings; a reduction in inspection costs is more
beneficial, the more likely it is that the consumer inspects relatively
many alternatives.

\subsection{Market Demand}

Using Theorem \ref{theo:effective_values}, it is straightforward
to derive a market demand function when heterogeneous consumers optimally
solve the SD problem. In particular, let the effective value $w_{ij}$
for each consumer $i$ be a realization of the random variable $W_{j}$
and gather the random variables in $\boldsymbol{W}=\left[W_{0},\dots,W_{\left|J\right|}\right]^{\prime}$.
For a unit mass of consumers the market demand for a product $j$
then is given by 
\begin{equation}
D_{j}=\mathbb{E}_{h}\left[\mathbb{P}_{\boldsymbol{W}}\left(W_{j}\geq W_{k}\forall k\in J\backslash j\right)\right]\label{eq:market_demand_general}
\end{equation}
where the expectations operator $\mathbb{E}_{h}\left[\cdot\right]$
integrates over all permutations of the order in which products are
discovered by a consumer.

As the effective value decreases in the position at which a product
is discovered, \eqref{eq:market_demand_general} reveals that the
demand for a product depends on the probability of each position at
which it is displayed. Specifically, the demand for a product exhibits
ranking effects; products that are more likely to be discovered early
are more likely to be bought. As discussed in detail in the next section,
this follows from the structure of the SD problem. As search progresses,
it becomes less likely that a consumer has not yet settled for an
alternative; hence, fewer consumers become aware of products that
would be revealed later, leading to a lower demand for such products.

\section{Comparison of Search Problems \label{sec:Comparison}}

To highlight implications of limited awareness and how the SD problem differs from existing approaches,
I compare it with the two classical sequential search problems; directed
search as in \citet{Weitzman1979} and random search as in \citet{McCall1970}.
Both these search problems are nested within the SD problem. Directed
search results if the consumer initially has full awareness (i.e.
$S_{0}=J$) such that the consumer 
knows all partial valuations prior to search and does not need to discover products. Random search results
if discovering a product reveals full information on this product, hence the consumer always both 
inspects and discovers a product, precluding him to use partial product information to 
only inspect promising products.\footnote{Directed
search also results if discovery costs are zero such that the consumer first discovers all products and only
then starts inspecting, whereas random search also results if inspection costs are zero and the consumer
inspects any products he discovers.}

For clarity, I focus the comparison on the case where products are
discovered one at a time ($n_{d}=1$) and where the consumer initially
only knows an outside option ($S_{0}=\emptyset$). Furthermore, valuations
$x_{j}$ and $y_{j}$ are assumed to be realizations of mutually independent
random variables $X$ and $Y$, where the consumer has rational expectations
such that beliefs are correct. Assumptions specific to each search
problem are described below. 

\textbf{Search and Discovery (SD): }The consumer searches as described
in Section \ref{sec:Model}, incurring inspection costs $c_{s}$ and
discovery costs $c_{d}$. Without loss of generality, I assume that
the consumer discovers products in increasing order of their index,
making subscripts for position $h$ and product $j$ interchangeable. 

\textbf{Random Search (RS):} When discovering a product $j$, the
consumer reveals both $x_{j}$ and $y_{j}$; hence does not have to
pay a cost to inspect the product. Costs to reveal this information
are given by $c^{RS}$. In this case, the consumer optimally stops
and buys product $j$ if $x_{j}+y_{j}\geq z^{RS}$. The reservation
value is given by $z^{RS}=\mu_{X}+\mu_{Y}+\tilde{\xi}$, where $\tilde{\xi}$
is the same as in \eqref{eq:res_search_explicit} but defined over
the joint distribution of demeaned $X$ and $Y$. Products are discovered
in the same order as in SD. Furthermore, I assume $u_{0}<z^{RS}$
to ensure a non-trivial case. 

\textbf{Directed Search (DS): }The consumer initially observes $x_{j}\forall j$,
based on which he chooses to search among alternatives following Weitzman's
(1979) reservation value policy. Costs to inspect product $j$ are
given by a function $c_{j}^{DS}=v_{DS}(c_{s},h_{j})$, where $c_{s}$
are baseline costs that are adjusted for the position through a function
$v_{DS}:\mathbb{R}_{+}^{2}\rightarrow\mathbb{R}_{+}$ which is assumed
to be strictly increasing in a product's position $h_{j}$. As costs
vary across products, reservation values are given by $z_{j}^{s}=x_{j}+\xi_{j}$,
where $\xi_{j}$ is the same as in \eqref{eq:res_search_explicit}
with product-specific inspection costs. The assumption on $v_{DS}(c_{s},h_{j})$
implies that $\xi_{j}$ decreases in $j$. I impose this functional
form restriction as otherwise the DS problem does not generate similar
patterns, as discussed in Section \ref{subsec:Ranking-Effects}. 

\subsection{Stopping decisions}

In search settings, consumers' stopping decisions determine which
products consumers consider and buy. Stopping decisions therefore
shape how firms compete in prices, quality or for being discovered
early during search. Hence, comparing stopping decisions across the
different search problems provides important insights on how well
existing approaches are able to capture the more general setting where
consumers are not aware of all alternatives and use partial information
to determine whether to inspect products. 

In the SD problem, a consumer always stops search at a product $k$
whenever the product is both promising enough to be inspected  and
offers a large enough valuation to not make it worthwhile to continue
discovering more products. Formally, this is given by the condition
$x_{k}+\min\{y_{k},\xi\}\geq z^{d}$. The probability that a consumer
will stop searching before discovering product $j$ therefore is given
by 
\begin{equation}
\mathbb{P}_{\boldsymbol{X},\boldsymbol{Y}}(X_{k}+\min\{Y_{k},\xi\}\leq z^{d}\forall k<j)=1-\mathbb{P}_{X,Y}(X+\min\left\{ Y,\xi\right\} \leq z^{d})^{j-1}\label{eq:stop_SD}
\end{equation}

Similarly, in the RS problem, a consumer will always stop search at
a product $k$ whenever $x_{k}+y_{k}\geq z^{RS}$, hence the probability
of stopping search before discovering product $j$ is given by 
\begin{equation}
\mathbb{P}_{\boldsymbol{X},\boldsymbol{Y}}(X_{k}+Y_{k}\leq z^{RS}\forall k<j)=1-\mathbb{P}_{X,Y}(X+Y\leq z^{RS})^{j-1}\label{eq:stop_RS}
\end{equation}

In both search problems, a consumer may stop search before discovering
a product $j$. Consequently, stopping decisions in the SD and the
RS problem imply the same feature: Products that a consumer initially
has no information on may never be discovered and bought, independent
of how the consumer values them. 

However, as the consumer has the option of not inspecting products
with low partial valuations, stopping probabilities differ. In
particular, in the case where the total cost to reveal all information
about a product are the same, stopping probabilities are smaller in
the SD problem. This is highlighted in Proposition \ref{prop: diff-RS-SD stopping}
and follows from the fact that not having to inspect alternatives
with small partial valuations allows to save on inspection costs.
This increases the expected benefit of discovering more products,
which implies a smaller probability of search stopping, and that on
average, more products will be discovered in the SD problem. 

\begin{proposition}
\label{prop: diff-RS-SD stopping}If costs in the RS problem are given
by $c^{RS}=c_{s}+c_{d}$, a consumer on average ends search at earlier
positions in the RS than in the SD problem. 
\end{proposition}

In contrast, stopping decisions are different in the DS problem. As
the consumer initially knows of the existence of all products and
can order them based on partial information, there is no stopping
decision in terms of discovering products. Instead, the consumer directly
compares all partial valuations and the different inspection costs,
based on which he decides the order in which to inspect products.
Hence, he can directly inspect highly valued products even when they
are presented at the last position. 

This difference arises from the different assumptions on consumers'
initial information and is paramount in the analysis of search frictions.
Consider an equilibrium setting where horizontally differentiated
alternatives are supplied by firms that compete by setting mean partial
valuations \citep[e.g. by setting prices as in][]{Choi2016a}. If
consumers are aware of all alternatives and search as in the DS problem,
all firms will compete directly with each other. In contrast, in a
SD problem, the firm that is discovered first initially competes only
with the option of discovering potentially better products.
This difference is further illustrated in Appendix \ref{sec:Sellers'-decisions},
and as it determines how firms compete, will lead to different equilibrium
dynamics.\footnote{To give an example, \citet{Anderson1999} and \citet{Choi2016a} model
a similar environment, with the difference that in the former, consumers
initially are not aware of any alternatives, whereas in the latter
they are aware and observe prices of all alternatives. Whereas in
the former, decreasing inspection costs lowers the equilibrium price
in a symmetric equilibrium, the opposite holds in the latter environment. }

\subsection{Ranking Effects\label{subsec:Ranking-Effects}}

The above analysis already suggests that the demand structure differs
across the three search problems. To provide further details, I focus
on a particular pattern that is generated by all three search problems:
Market demand for a product decreases in its position. Such ranking
effects are important as they determine how fiercely sellers compete
for their products to be revealed on early positions, for example
through informative advertising or position auctions \citep[e.g.][]{Athey2011}.
Furthermore, they have received considerable attention in the marketing
literature, which has produced ample empirical evidence that suggests
their importance in online markets \citep[e.g.][]{Ghose2014,LosSantos2017,Ursu2018}. 

To compare the mechanism producing ranking effects across the search
problems, I use the following definition: The ranking effect for a
product is the difference in market demand of the product being revealed
at position $h$ and at $h+1$, with the corresponding exchange of
the product previously revealed at position $h+1$. Formally, this is given
by
\begin{equation}
r_{k}(h)\equiv d_{k}(h)-d_{k}(h+1)\label{eq:def-general-ranking-effect}
\end{equation}
where $d_{k}(h)$ denotes the market demand for a product when revealed
at position $h$ in search problem $k\in\{SD,RS,DS\}$. For clarity,
product specific subscripts are either omitted or exchanged with position
subscripts in the following. The former is feasible as effective values are
assumed to be independent realizations of a random variable $W$.

To investigate ranking effects, it is first necessary to derive the
market demand at a particular position $h$. For a unit mass of consumers
with independent realizations of effective values, it is given by
\begin{multline}
d_{SD}(h)=\mathbb{P}_{W}(W<z^{d})^{h-1}\biggl[\mathbb{P}_{W}(W\geq z^{d})\\
+\mathbb{P}_{W}(W<z^{d})^{\left|J\right|-(h-1)}\mathbb{P}_{\boldsymbol{W}}(W\geq\max_{k\in J}W_{k}|W_{k}<z^{d}\forall j)\biggr]
\end{multline}

The expression follows from Theorem \ref{theo:effective_values} which
implies that if a consumer discovers a product with $w_{j}\geq z^{d}$,
he will stop searching and buy a product $j$. The consumer will only
discover and have the option to buy a product on position $h$ if
$w_{j}<z^{d}$ for all products on better positions. In contrast,
when $w_{j}<z^{d}$, the consumer will first discover more products,
and only recall $j$ if he discovers all products and $j$ is the
best among them.

In the latter case, a product's position does not affect market demand;
once all products are discovered, products are equivalent in terms
of their inspection costs and the order in which they are inspected
is only determined based on partial valuations. This implies that
the ranking effect in the SD problem is independent of the number
of alternatives and simplifies to
\begin{equation}
r_{SD}(h)=\mathbb{P}_{W}\left(W\geq z^{d}\right)\left[\mathbb{P}_{W}(W<z^{d})^{h-1}-\mathbb{P}_{W}(W<z^{d})^{h}\right]\label{eq:ranking-eff-sd}
\end{equation}

This expression reveals that the ranking effect in the SD problem
solely results from the difference in the probability of a consumer
reaching positions $h$ or $h+1$ respectively. Besides the distribution
of valuations and the inspection and discovery costs, Proposition
\ref{prop: ranking-sd-rs} shows that the ranking effect is determined
by the position $h$ to which the product is moved. When $h$ is large,
fewer consumers will not have already stopped searching before  reaching
$h$. Hence, the later a product is revealed, the smaller is the increase
in demand when moving one position ahead.

The demand in a random search problem is derived similarly. In RS,
a consumer will only be able to buy a product if he has not stopped
searching before, which requires that $x+y<z^{RS}$ for all products
on better positions. Furthermore, a consumer will also only recall
a product if he has inspected all alternatives. Similar to the SD
problem, this implies that the ranking effect in the RS problem is
given by
\begin{equation}
r_{RS}(h)=\mathbb{P}_{X,Y}\left(X+Y\geq z^{RS}\right)\left[\mathbb{P}_{X,Y}\left(X+Y<z^{RS}\right)^{h-1}-\mathbb{P}_{X,Y}\left(X+Y<z^{RS}\right)^{h}\right]\label{eq:ranking-eff-rs}
\end{equation}

Comparing \eqref{eq:ranking-eff-sd} with \eqref{eq:ranking-eff-rs}
reveals that ranking effects in the RS problem are produced by the
same mechanism as in the SD problem. In both search problems; fewer
consumers buy products at later positions due to the increasing the
probability of having stopped searching before discovering these products.
It follows that in both search problems, ranking effects decrease
in the position and are independent of the total number of alternatives. 

Though their extent generally differs, Proposition \ref{prop: ranking-sd-rs}
additionally shows that at later positions, ranking effects will be
larger in the SD problem. The result is a direct implication of Proposition
\ref{prop: diff-RS-SD stopping}; as a consumer is more likely to
reach a product at a later position in the SD problem, ranking effects
at later positions will be larger. 
\begin{proposition}
\label{prop: ranking-sd-rs}The ranking effect in both the SD and
the RS problem decreases in position $h$ and is independent of the
number of alternatives. Furthermore, if $c^{RS}=c_{s}+c_{d}$, there
exists a threshold $h^{*}$ such that $r_{SD}(h)\geq r_{RS}(h)$ for
all $h>h^{*}$.
\end{proposition}
Given the different stopping decisions, ranking effects in directed
search do not result from consumers having stopped searching before
reaching products revealed at later positions. Instead, they result
from differences in the cost of inspecting products at different positions.
To see this, write the ranking effect in the DS problem as\footnote{Alternatively, 
ranking effects could be modeled in a DS problem by assuming that the consumer
initially has full information on some products. In this case, 
the model effectively has only 2 positions (full and partial information), 
and hence would not be able to explain the decrease in demand across all positions resulting 
from the SD problem. } 
\begin{align}
r_{DS}(h) & =\mathbb{E}_{\tilde{W}_{h}}\biggl[\prod_{k\neq h}\mathbb{P}(\tilde{W}_{k}\leq\tilde{W}_{h})\biggr]-\mathbb{E}_{\tilde{W}_{h+1}}\biggl[\prod_{k\neq h+1}\mathbb{P}(\tilde{W}_{k}\leq\tilde{W}_{h+1})\biggr]\label{eq:ranking-eff-ds}
\end{align}

This expression reveals that the ranking effect results from two sources
in the DS problem. First, by moving a product $j$ one position ahead,
the product previously on position $h$ is now more costly to inspect,
making it more likely that $j$ is bought for any $\tilde{w}_{j}$.
Second, by making it less costly to inspect $j$, the distribution
of $\tilde{w}_{j}$ shifts such that larger values $\tilde{w}_{j}$
become more likely. 

In contrast to RS and SD, the ranking effect in the DS problem depends
on the number of available alternatives. In RS and SD, ranking effects
result from the decreasing probability of a consumer having stopped
searching before reaching a particular position, which does not depend
on how many alternatives there are in total. In DS, however, a consumer
directly compares all alternatives based on partial valuations. Adding
more alternatives thus will affect the demand on each position.

Specifically, Proposition \ref{prop:ranking-ds} shows that ranking
effects in the DS problem will be smaller if there are many alternatives.
The reason is that as the number of alternatives increases, each product
is less likely to be bought and differences in the position-specific
market demand decrease. Note, however, that in cases where the probability
of consumers buying products on the last positions is very small or
exactly zero (e.g. when inspection costs are large), adding more alternatives
will not affect ranking effects in the DS problem. 
\begin{proposition}
\label{prop:ranking-ds}The ranking effect in the DS problem is weakly
decreasing in the number of alternatives. 
\end{proposition}
A second difference to the RS and SD problems is that the ranking effect
does not necessarily decrease in position. This is possible as there
are two counteracting channels through which position affects the
ranking effect in a DS problem. First, as there is lower demand for
products at later positions, differences between them will be smaller.
Second, if $v_{DS}(c_{s},h)$ is such that $\xi_{h}$ decreases in
$h$ at an increasing rate, the difference in the purchase probability
at $h$ instead of at $h+1$ increases in the position. When the latter
dominates, the ranking effect will first increase in position.

The above comparison highlights that the mechanism producing ranking
effects in the DS problem is distinct from the one in the SD and RS
problems, leading to a different demand structure. In the former, ranking
effects result from differences in inspection costs relative to differences
in partial valuations. Hence, a better partial valuation is a substitute
for moving positions ahead. In contrast, in a SD or RS problem, a
product's large partial valuation does not affect consumers that stop
search before discovering it. Hence, offering a larger partial valuation
does not substitute for being discovered early in a SD or RS problem.\footnote{Note, however, that in an equilibrium setting, offering larger partial
valuations may indirectly serve as a substitute for being discovered
early by raising consumers' expectations and induce them to search
longer.} 

Moreover, the size of ranking effects determines how important
it is for products to be revealed on an early position. As ranking
effects are independent of the number of alternatives in SD and RS,
so are sellers' incentives to have their products revealed early during
search. In contrast, in DS, the demand increase of moving positions
ahead becomes smaller when the number of alternatives increases. Hence,
sellers can have smaller incentives to be revealed on early positions
when there are many, relative to when there are only few alternatives. 

Finally, the above comparison between the number of alternatives and
ranking effects also suggests the existence of an empirical test to
distinguish the search modes in some settings. If data is available
that allows to test whether ranking effects depend on the number of
alternatives, then it will be possible to empirically determine whether
a DS problem, instead of a RS or SD problem provides a framework that
better captures ranking effects in a particular setting. Furthermore,
if data is available that allows to test whether a product's partial
valuation has an effect on whether it is inspected, it will be possible
to distinguish between RS and SD. 

\subsection{Expected Payoff}

If costs are specified such that the total costs of revealing all
product information remain the same, then the three search problems
differ only in the information the consumer can use during search.
A comparison of a consumer's expected payoff based on such a specification
therefore provides some insight into whether it is always to the consumer's
benefit to provide information that helps to direct search towards
some alternatives. 

For total costs of revealing full information about a product on position
$h$ to be the same in the three search problems, inspection costs
in the RS and DS problem are specified as $c^{RS}=c_{s}+c_{d}$ and
$c_{j}^{DS}=c_{s}+h_{j}c_{d}$ respectively. 

The SD problem extends the RS problem by additionally providing the
consumer with the option to not inspect products depending on their
partial valuations. This allows the consumer to save on inspection
costs by not inspecting products with small partial valuations. As
stated in Proposition \ref{prop:sd-rs-diff-exp-payoff}, this increases
the expected payoff which implies that providing product information
across two layers, as done for example by online retailers or search
intermediaries, is beneficial for consumers. 
\begin{proposition}
\label{prop:sd-rs-diff-exp-payoff}If $c^{RS}=c_{s}+c_{d}$, then
a consumer's expected payoff in the SD problem is larger than in the
RS problem. 
\end{proposition}
In contrast to the SD problem, the consumer can use all partial valuations
to direct search in the DS problem. Hence, if inspection costs for
all products are the same in both problems (i.e. $c_{j}^{DS}=c_{s}\forall j$),
a consumer will have a larger expected payoff in the DS problem as he can directly
inspect products with large partial valuations. However,
under the assumption that total costs of revealing full information
are the same in both search problems, a more detailed analysis is
necessary to determine which search problem offers a larger expected
payoff. 

Denote a consumer's expected payoff in a search problem $k$ as $\pi_{k}$
for $k\in\left\{ SD,DS\right\} $. Proposition \ref{prop:Expected-consumer-welfare}
implies that expected payoffs are given by 
\begin{align*}
p_{SD} & =\mathbb{E}_{\hat{\boldsymbol{W}}}\left[\max\{u_{0},\max_{j\in J}\hat{W}_{j}\}\right]\\
p_{DS} & =\mathbb{E}_{\tilde{\boldsymbol{W}}}\left[\max\{u_{0},\max_{j\in J}\tilde{W}_{j}\}\right]
\end{align*}
Furthermore, let $H_{k}(\cdot)$ denote the cumulative density of
the respective maximum value over which the expectation operator integrates
in problem $k$. The difference in expected payoffs of the SD and
the DS problem then is given by
\begin{equation}
p_{SD}-p_{DS}=\int_{z^{d}}^{\infty}H_{DS}(w)-1\text{d}w+\int_{u_{0}}^{z^{d}}H_{DS}(w)-H_{SD}(w)\text{d}w\label{eq:diff-sd-ds-payoff}
\end{equation}
The first expression in \eqref{eq:diff-sd-ds-payoff} is negative,
capturing the advantage of observing partial valuations for all products
and being able to directly inspect a product at a later position.
Given $H_{DS}(w)\leq H_{SD}(w)$ on $w\in\left[u_{0},z^{d}\right]$,
the second expression in \eqref{eq:diff-sd-ds-payoff} is positive,
revealing that directly observing all partial valuations $x_{j}$
does not only yield benefits. 

The latter stems from the difference in how inspection and discovery
costs are taken into consideration in the two dynamic decision processes.
In DS, the total cost of inspecting a product $j$ at a later position
is directly weighed against its benefits given the partial valuations.
In contrast, in SD, the consumer first weighs the discovery costs
against the expected benefits of discovering a product with a larger
partial valuation. Once product $j$ is revealed, the accumulated
cost paid to discover $j$ ($jc_{d}$) is a sunk cost and does not
affect the decision whether to inspect $j$.

Hence, in cases where products on early positions have below-average
partial valuations $x_{j}$, the optimal policy in SD tends to less
often prescribe to inspect these products compared to the direct cost
comparison in DS. In some cases, the former can be more beneficial,
leading to a larger expected payoff.\footnote{For example, this is the case if $X\sim N\left(0,\frac{1}{3}\right)$,
$Y\sim N\left(0,\frac{2}{3}\right)$, $c_{s}=c_{d}=0.05$ and $\left|J\right|=10$.} Directly revealing all 
partial valuations therefore does not always
improve a consumer's benefit, if the consumer continues to incur the
same total costs to reveal the full valuation of any given product.
\footnote{No threshold result as in Proposition \ref{prop:welfare-change} applies
in this case. The first expression in \eqref{eq:diff-sd-ds-payoff}
decreases whereas the second expression increases in the number of
alternatives. }

\subsection{Empirical Implications}

Differences in the underlying search problem also have implications 
for the estimation of structural search models. For example, a structural search model will use 
price differences across all products to inform parameter estimates if it abstracts from limited 
awareness and assumes that consumers observe all prices prior to search. Consumers not inspecting 
low-price products they are unaware of then may be spuriously attributed either to a small price 
sensitivity or large inspection costs. Whereas there are many applications of structural search models and 
an ubiquity of settings where consumers remain unaware of some alternatives, the sensitivity of results 
from structural search models to limited awareness remains unclear. 

I therefore investigate the implications of estimating either a random or 
directed search model in a setting where
consumers instead solve the search and discovery problem. I focus on a scenario 
where preference and cost parameters 
are estimated using data on consumers' consideration sets and purchases; a common case
as consideration sets are observable in click-stream or survey data.
Using a simple specification,\footnote{The empirical literature extends the simple specification for a range
of settings, for example by introducing heterogeneous preferences. The main rationale continues to
hold in such settings.} I first analyze how the different models attribute observed stopping
decisions to structural parameters. A numerical exercise then reveals that this can lead to sizable 
differences in parameter estimates and counterfactual predictions. 

\textbf{Empirical setting}: The data consist of consumers' consideration
sets and purchases,\footnote{The simulated data also contains consumers with an empty consideration set, i.e. those that 
did not search any alternatives. This corresponds to an ideal setting where the whole population of consumers is observed.} 
as well as a number of characteristics for each
of the available products. The utility of purchasing product $j$ is specified as $u_{j}=\boldsymbol{x}_{j}'\beta+y_{j}$,
where \textbf{$\boldsymbol{x}_{j}$ }is a vector containing the observed
product characteristics, $\beta$ is a vector of preference parameters
and $y_{j}$ is an idiosyncratic unobservable taste shock with mean
zero. Depending on the model, consumers are assumed to reveal $\boldsymbol{x}_{j}$
when either discovering $j$ (SD) or inspecting $j$ (RS), or know
$\boldsymbol{x}_{j}$ prior to search (DS). $y_{j}$ is revealed after
inspecting $j$ in all three models. 

Given this setting, Table \ref{tab:purch_cond} shows sufficient or
necessary conditions for the purchase of product $j$ across the three
models, conditional on $j$ being the best product inspected and (ii) 
the observed consideration set not coinciding with the set of all available alternatives. The condition
for the SD problem shows that a purchase of product $j$ can be independent
of realized valuations of products that the consumer is not aware
of in the purchase period ${\bar{t}}$.\footnote{The condition is sufficient but not necessary. A
lternatively, the
consumer can first become aware of all alternatives, before then purchasing
$j$. } 
$j$ only needs to offer ``good enough'' characteristics relative to
the mean and to products the consumer is aware of at the time of purchase. The RS model features the same structure; a consumer will
end search and buy product $j$ if its valuation exceeds the reservation
value. However, $\Xi$ and $\tilde{\xi}$ depend differently on the
underlying costs and distributions of characteristics in $\boldsymbol{x}_{j}$
and $y_{j}$. Through these non-linear functions, a RS model will
attribute observed limited consideration sets differently to preference
and cost parameters. 

In the DS model rationalizing the purchase of $j$ requires that the
valuation of the purchased product is larger than the search values
of all uninspected products. If, for example, $\boldsymbol{x}_{k}>\boldsymbol{x}_{j}$
for an uninspected product, the DS model will require either relatively
small preference parameters, or relatively large inspection costs.
Hence, depending on the characteristics of the uninspected products,
rationalizing limited consideration sets in a DS model will require
a combination of large inspection costs and attenuated preference
parameters, as the estimation procedure will try to fit an inequality
for each uninspected product. 

\begin{table}[htb] \centering
\caption{Purchase conditions} \label{tab:purch_cond}
\begin{threeparttable} \footnotesize
\begin{tabular}{ccccllr} 
\midrule
\emph{SD} && $ (\boldsymbol{x}_{j}-\mu_{\boldsymbol{X}})'\beta + y_j $&$ \geq \Xi $ & \& & $ (\boldsymbol{x}_{j}-\boldsymbol{x}_{k})'\beta + y_j \geq \xi_k \forall k \in A_{\bar{t}} $ & (sufficient) \\ 
\emph{RS} && $ (\boldsymbol{x}_{j} - \mu_{\boldsymbol{X}})' \beta + y_j $ & $ \geq \tilde{\xi} $ &&&(necessary)\\
\emph{DS} && $ (\boldsymbol{x}_{j}-\boldsymbol{x}_{k})'\beta + y_j $& $\geq \xi_k \forall k \notin C_{\bar{t}} $ &&& (necessary)\\
\midrule 
\end{tabular}
\begin{tablenotes}
\item \footnotesize{\emph{Notes:} Sufficient or necessary conditions for purchase of product $j$ conditional on $u_j \geq u_k \forall k\in C_{\bar{t}}$ and 
$J \nsubseteq C_{\bar{t}}$. $\bar{t}$ denotes the purchase period.}
\end{tablenotes}
\end{threeparttable}
\end{table}

To investigate the extent to which this influences results from structural search models, 
I perform simulations for this
setting. First, I simulate consumers solving a SD problem with the
given utility specification and under the assumption that consumers initially
aware of one product. Using these data, I then estimate structural
parameters in search models based on either the RS and DS problem. 
For the DS problem, two specifications are estimated. DS1 is a baseline
where inspection costs are parameterized as $c_j^{DS1} = c_s$. DS2 introduces an additional
cost parameter such that inspection costs increase in position $h_j$ with 
$c_j^{DS2} = c_s + c_d h_j$. This specification
additionally uses data on the order in which products are discovered by consumers.
For all three models, the estimation fits inequalities based on the conditions
of Table \ref{tab:purch_cond}, as well as other inequalities coming
from continuation and purchase decisions. Details on the maximum likelihood estimation are provided in the appendix.
As comparison, I also present estimates of a full information (FI)
model. 

Results of such a simulation are presented in Table \ref{tab:coef_estim}. 
Parameters for this particular simulation are shown in the same table and were chosen to reflect a setting with relatively
few searches, as is often the case in click-stream data.\footnote{For example, 
\cite{Ursu2018} reports an average of 1.12 clicks per consumer and two thirds of consumers ending up booking a hotel.\cite{Chen2016} 
reports an average of 2.3 clicks per consumer using data only
on consumers that ended up booking a hotel.} 
To account for the fact that assuming the distribution of $y_{j}$
is a normalization in the empirical context, estimates are presented
as a ratio to the coefficient of the second characteristic. Given its 
negative coefficient, this characteristic will be interpreted as a product's price. 

\begin{table}[htb] \centering 
\caption{Estimated Coefficients and Search Set Size} \label{tab:coef_estim}
\begin{threeparttable}
\begin{tabular}{ccccccccc} 
\midrule
&& \#Searches& Purchases (\%)& $ \beta_2 $& $ \beta_1 / | \beta_2 |$& $ \beta_3 / | \beta_2 | $ & $ c_s / | \beta_2 | $& $ c_d / | \beta_2 |$\\
\midrule
\emph{SD}&&1.35&63.70&-1.00&1.00&3.50&0.03&0.06\\
\emph{DS1}&&1.18&65.48&-0.19&1.01&2.58&1.79& \\
\emph{DS2}&&1.18&65.22&-0.19&1.01&2.72&1.58&0.01\\
\emph{RS}&&1.00&72.85&-0.82&1.28&5.21&0.05& \\
\emph{FI}&& &60.54&-0.62&1.00&5.01& & \\
\midrule 
\end{tabular}
\begin{tablenotes}
\item \footnotesize{\emph{Notes:} Estimation from a simulated dataset with 2,000 consumers and 30 products per consumer. 
Characteristics are independent draws (across consumers and products) from 
$x_{1j} \sim N(2,3.0)$, $x_{2j} \sim N(3.5,1.0)$ and $y_j \sim N(0,1)$. The third characteristic is an outside dummy.
The data is generated based on the \emph{SD} model with $n_d=|A_0|=1$, with parameters in the 
estimated models denoted by $c^{RS}=c_s$, $c_j^{DS1} = c_s$ and $c_J^{DS2} = c_s + c_d h_j$.
The first two columns are based either 
on the generated data (SD) or estimated by generating 5,000 search paths for each consumer.  }
\end{tablenotes} 
\end{threeparttable}
\end{table}

The results show that both DS specifications are able to match the number of purchases, as 
well as the relative preference coefficients relatively well, despite the price coefficient being
strongly attenuated and the number of searches being underestimated. However, inspection costs are strongly accentuated in both DS models. 
This offers a novel explanation for the large estimates 
of baseline costs estimated with some DS models \citep[e.g.][]{Chen2016,Ursu2018}: 
By not accounting for consumers not being aware of some alternatives, a DS model 
spuriously attributes consumers not inspecting products they are not aware of 
to large inspection costs.\footnote{Other explanations for large search search 
cost estimates are incomplete search histories \citep[e.g.][]{Ursu2018} and 
heterogeneous prior beliefs \citep[][]{Jindal2020}. } This continues to occur in 
the DS2 model that could rationalize ranking effects produced by the SD model through
inspection costs that increase in the position at which a product is discovered. However, 
the results show that instead the DS2 model estimates only a small increase in
inspection costs across positions and also strongly overestimates baseline inspection costs. 

The RS model underestimates inspection costs; they are less than 
the combined inspection and discovery costs. Moreover, the ratio of preference parameters deviates from the
true values. The large differences in the estimated coefficient for
the outside option result from how the different models interpret
consumers not inspecting or not buying. Whereas in the DS problem
this occurs from large inspection costs, the RS
model attributes the lack of search mainly to a good outside option. 

Differences in the structural search models also influence results from counterfactual simulations. Table \ref{tab:counterfactuals}
shows the results of two different counterfactuals for each of the models. For each 
counterfactual scenario, parameters from Table \ref{tab:coef_estim} are used for each model to simulate 
consumer surplus ($CS$) and the demand for the outside option ($D_0$), as well as for products shown on the first ($D_1$) and fifth ($D_5$) 
position. Throughout, results are expressed in percentage deviations from the baseline scenario. 

The first counterfactual consists of removing all search costs, which can be used to 
gauge the effects of removing the search friction. For both DS models, accentuated baseline inspection costs
lead to a larger increase in consumer surplus compared to the SD model with which 
the data was generated. Moreover, removing costs in the DS models makes consumers more likely to purchase any 
product, independent of their position. In contrast, demand in the SD and RS model decreases for both products listed.
This stems from the inherent ranking effects where products on early positions are bought 
more frequently as consumers stop search early. In this case, removing all costs moves demand to later positions.

The second counterfactual scenario analyzes the effects of a one percent price decrease of products discovered on the fifth position. 
The change in demand for the first product highlights an important difference in the substitution pattern. 
In the data-generating SD model, the demand for products on the first position decreases by little, as the price decrease of a product 
on a later position does not affect choices of consumers that stopped search before becoming aware of the product. 
In contrast, in a DS (or FI) model, consumers who were previously buying products on the first 
few positions observe the price decrease and can directly substitute to 
the fifth product. This translates into more substitution from the first few positions as 
a response to a price decrease of a product on a later position. The predicted
changes in the demand for the fifth product further highlight that the different models 
lead to different predictions for consumers' responses to price changes; whereas the DS1 
and RS models underestimate, the DS2 model overestimates the increase in demand in response to 
the price change. 
 
\begin{table}[htb] \centering 
\caption{Counterfactuals} \label{tab:counterfactuals}
\begin{threeparttable}
\begin{tabular}{cccccccc} 
\midrule
&& \multicolumn{3}{c}{Remove costs} & \multicolumn{3}{c}{$\Delta p_5 = -1\%$} \\ 
\cline{3-5} \cline{6-8}
&& $\Delta CS $& $\Delta D_1$& $\Delta D_5 $& $\Delta CS $& $\Delta D_1$& $\Delta D_5 $\\
\midrule
\emph{SD}&&28.60&-37.35&-2.32&0.02&-0.01&1.81\\
\emph{DS1}&&85.06&38.04&43.11&0.01&-0.04&1.72\\
\emph{DS2}&&81.38&15.53&29.19&0.01&-0.03&2.75\\
\emph{RS}&&18.73&-25.36&-11.78&0.01&-0.02&1.49\\
\emph{FI}&&0.00&0.00&0.00&0.01&-0.05&1.91\\
\midrule 
\end{tabular}
\begin{tablenotes}
\item \footnotesize{\emph{Notes:} Results from simulated counterfactuals based on Table \ref{tab:coef_estim}, where (i) all costs are set to zero and (ii) the price for the 5th 
product is reduced by 1 \% for each consumer. All changes are expressed in \% relative to the baseline. Demand and consumer
surplus are calculated by averaging across 5,000 simulated search paths for each consumer. }
\end{tablenotes} 
\end{threeparttable}
\end{table}

Though results from only a single simulation are presented, I obtained
qualitatively similar results across a wide range of parameter values.\footnote{These results can be replicated with the supplementary material. }
Throughout, DS models overestimate inspection costs and all estimated models can lead
to sizable differences in parameters and results from counterfactual predictions. 
Nonetheless, the SD problem will be more similar to the DS problem if consumers
are aware of many alternatives when they end search (e.g. due to small discovery costs). 
Similarly, if consumers inspect most products they discover independent of their characteristics,
the SD problem will be more similar to the RS problem. When estimating
search models, researchers should therefore carefully consider the
degree to which limited awareness plays a role in the specific setting
they are studying and which model is appropriate.

To this end, the results of Propositions \ref{prop: ranking-sd-rs} and \ref{prop:ranking-ds} 
can be used to empirically differentiate the search modes in some settings.
If data are available that allow to test whether ranking effects depend on the number of
alternatives, it will be possible to empirically determine whether
a DS problem, instead of a RS or SD problem provides a framework that
better captures ranking effects. Furthermore,
if data are available that allow to test whether a product's partial
valuation has an effect on whether it is inspected, it will be possible
to distinguish between RS and SD.

\section{Conclusion \label{sec:Discussion}}

This paper introduces a search problem that generalizes existing frameworks
to settings where consumers have limited awareness and first need
to become aware of alternatives before being able to search among
them. The paper's contribution is to provide a tractable solution
for optimal search decisions and expected outcomes for this search
and discovery problem. Moreover, a comparison with classical random
and directed search highlights how limited awareness and the availability
of partial product information determine search outcomes and expected
payoffs. 

A promising avenue for future research is to build on this paper's
results and study limited awareness in an equilibrium setting. This
could yield novel insights into how consumers' limited information
shapes price competition. Furthermore, the search and discovery problem
can serve as a framework to analyze how firms compete for consumers'
awareness. For example, informative advertising can make it more likely
that consumers are aware of a seller's products from the outset. Ranking
effects derived in this paper already suggest that it will be in a
seller's best interest to make consumers aware of his product, but
further research is needed to determine equilibrium dynamics. 

Another avenue for future research entails incorporating the search
and discovery problem into a structural model that is estimated with
click-stream data. The available actions in the search and discovery
problem closely match how consumers scroll through product lists (discovery)
and click on products (inspection) on websites of search intermediaries
and online retailers. By accounting for the fact that consumers initially
do not observe entire list pages, such a model could improve the estimation
of consumers' preferences, inspection costs and ranking effects relative
to models that abstract from consumers not observing the whole product
list. 

\newpage
\singlespacing 

\bibliographystyle{chicago}
\bibliography{library_fixed}

\newpage 

\appendix 
\section*{Appendix}
\onehalfspacing
\rhead{}
\section{Proofs of main Theorems and Propositions}\label{sec:Proofs}

\subsection{Theorem \ref{theo:optimal_policy}\label{subsec:Proof-Theorem}}
Let $\Theta(\Omega_{t},A_{t},z)$ denote the value function of an
alternative decision problem, where in addition to the available actions
in $A_{t}$, there exists a hypothetical outside option offering value
$z$. As the SD problem satisfies that taking an action does not change
the state of another available action and has the same branching structure,
Theorem 1 of \citet{Keller2003} states that a Gittins index policy
is optimal and that the following holds:\footnote{Compared to the baseline branching framework discussed by \citet{Keller2003},
the SD problem does not have discounting, and purchasing a product
is a ``terminal'' action. Note also that whereas not explicitly
stated by the authors, their framework accommodates the case where it is not known ex
ante to how many ``children'' an available action branches into.
This will be the case in the SD problem if the consumer does not
know the number of products he will discover. } 
\begin{equation}
\Theta(\Omega_{t},A_{t},z)=b-\int_{z}^{b}\Pi_{a\in A_{t}}\frac{\partial\Theta\left(\Omega_{t},\left\{ a\right\} ,w\right)}{\partial w}\mathrm{d}w\label{eq:conjecture}
\end{equation}
where $b$ is some finite upper bound of the expected immediate rewards.\footnote{Expected immediate rewards are in $ \left[- \max\{c_s,c_d\},\mathbb{E}[X+Y] \right] $, hence 
assuming finite mean of $X$ and $Y$ guarantees that they have a finite upper bound. } The Gittins index of action $d$ (discovering products) is defined
by $g_{t}^{d}=\mathbb{E}_{\boldsymbol{X}}\left[\Theta(\Omega_{t+1},A_{t+1}\backslash A_{t},g_{t}^{d})\right]$.
Suppose the consumer knows the total number of alternatives $\left|J\right|$,
and consider a period $t$ in which more discoveries will still be
available in $t+1$ with certainty. In this case we have 
\begin{align}
g_{t}^{d} & =\mathbb{E}_{\boldsymbol{X}}\left[\Theta(\Omega_{t+1},\{d,s1,\dots,sn_{d}\},g_{t}^{d})\right]-c_{d}\label{eq:def_gittins}\\
 & =\mathbb{E}_{\boldsymbol{X}}\left[b-\int_{g_{t}^{d}}^{b}\frac{\partial\Theta\left(\Omega_{t+1},\left\{ d\right\} ,w\right)}{\partial w}\prod_{k=1}^{n_{d}}\frac{\partial\Theta\left(\Omega_{t+1},\left\{ s_{k}\right\} ,w\right)}{\partial w}\mathrm{d}w\right]-c_{d}\nonumber 
\end{align}
where $s_{k}\in S_{t+1}\backslash S_{t}\forall k$. $\Theta(\Omega_{t},\left\{ s_{k}\right\} ,z)$
is the value of a search problem with an outside option offering $z$
and the option of inspecting product $k$ (with known partial valuation
$x_{k}$). $\Theta\left(\Omega_{t+1},\left\{ d\right\} ,w\right)$
is the value of a search problem with an outside option offering $z$,
and the option to discover more products. Finally, $\mathbb{E}_{\boldsymbol{X}}\left[\cdot\right]$
is the expectation operator integrating over the beliefs over the
$n_{d}$ random variables in $\boldsymbol{X}=\left[X,\dots,X\right]$,
which does not depend on time.

Optimality of the Gittins index policy then implies that when $z\geq g_{t+1}^{d}$,
the consumer will choose the outside option in $t+1$. Hence $\Theta\left(\Omega_{t},\left\{ d\right\} ,w\right)=w\forall w\geq g_{t+1}^{d}$
which yields $\frac{\partial\Theta\left(\Omega_{t},\left\{ d\right\} ,w\right)}{\partial w}=1\forall w\geq g_{t+1}^{d}$.
This implies that for $g_{t}^{d}\geq g_{t+1}^{d}$, $g_{t}^{d}$ does
not depend on whether more products can be discovered in the future,
 and the optimal policy is independent of the beliefs over the number
of available alternatives. As a result, as long as the Gittins index
is weakly decreasing during search, i.e. $g_{t}^{d}\geq g_{t+1}^{d}\forall t$,
it is independent of the availability of future discoveries and beliefs $q$. 

It remains to show that $g_{t}^{d}\geq g_{t+1}^{d}\forall t$ holds
in the proposed search problem. When $\left|J\right|=\infty$, $g_{t}^{d}=g_{t+1}^{d}$
is immediately given by the fact that in both periods infinitely many products remain
to be discovered and that the consumer has stationary beliefs (i.e. $q$ is constant and 
valuations are independent and identically distributed). For $\left|J\right|<\infty$,
backwards induction yields that this condition holds: Suppose that
in period $t+1$, no discovery action is available as all products have been discovered. In this case,
the Gittins index is given by
\begin{equation}
g_{t+1}^{d}=\mathbb{E}_{\boldsymbol{X}}\left[b-\int_{g_{t+1}^{d}}^{b}\prod_{k=1}^{n_{d}}\frac{\partial\Theta\left(\Omega_{t+1},\left\{ s_{k}\right\} ,w\right)}{\partial w}\mathrm{d}w\right]-c_{d}
\end{equation}
As $0\leq\frac{\partial\Theta\left(\Omega_{t},\left\{ d\right\} ,w\right)}{\partial w}\leq1$
and $\frac{\partial\Theta\left(\Omega_{t},\left\{ s_{k}\right\} ,w\right)}{\partial w}\geq0$,
it holds that 
{\footnotesize 
\begin{multline} 
\mathbb{E}_{\boldsymbol{X}}\left[b-\int_{g_{t+1}^{d}}^{b}\prod_{k=1}^{n_{d}}\frac{\partial\Theta\left(\Omega_{t+1},\left\{ s_{k}\right\} ,w\right)}{\partial w}\mathrm{d}w\right]\leq q \mathbb{E}_{\boldsymbol{X}}\left[b-\int_{g_{t}^{d}}^{b}\prod_{k=1}^{n_{d}}\frac{\partial\Theta\left(\Omega_{t+1},\left\{ s_{k}\right\} ,w\right)}{\partial w}\mathrm{d}w\right] + \\
(1-q)\mathbb{E}_{\boldsymbol{X}}\left[b-\int_{g_{t}^{d}}^{b}\frac{\partial\Theta\left(\Omega_{t},\left\{ d\right\} ,w\right)}{\partial w}\prod_{k=1}^{n_{d}}\frac{\partial\Theta\left(\Omega_{t},\left\{ s_{k}\right\} ,w\right)}{\partial w}\mathrm{d}w\right]
\end{multline}
}
which implies $g_{t}\geq g_{t+1}$. 

Finally, $\Theta(\Omega_{t+1},\{d,s1,\dots,sn_{d}\},g_{t}^{d})=V\left(\left\langle \bar{\Omega},\omega(x,z)\right\rangle ,\left\{ b0,s,\dots,sn_{d}\right\} ;\tilde{\pi})\right)$
in \eqref{eq:exp_net_ben_disc} implies $z^{d}=g_{t}^{d}$. Similarly,
the definition of the inspection and purchase values (in \eqref{eq:res_search}
and \eqref{eq:res_exp}) are equivalent to the definition of Gittins
index values for these actions and it follows that the reservation
value policy is the Gittins index policy.

\subsection{Theorem \ref{theo:effective_values} } \label{subsec:proof-theorem-effvalues}
\begin{proof}
As a product always is bought, it suffices to show that the optimal
policy never prescribes to buy product $j$ if there exists another
product $k$ with $w_{k}>w_{j}$. To account for the case where $C_{0}\neq \emptyset$,
define $z_{k}^{s}=\infty\forall k\in C_{0}$ which implies $\tilde{w}_{k}\equiv\min\left\{ z_{k}^{s},z_{k}^{b}\right\} =z_{k}^{b}\forall k\in C_{0}$.
First, consider the case where $k$ is revealed before $j$ ($h_{0}\leq h_{k}<h_{j}$).
In this case, $w_{k}>w_{j}$ if and only if either (i) $\tilde{w}_{k}\geq z^{d}$
or (ii) $z^{d}>\tilde{w}_{k}>\tilde{w}_{j}$. In the former, the optimal
policy prescribes to not discover products beyond $k$, hence not
to buy product $j$. This follows as $z_{k}^{s}\geq z^{d}$ and $z_{k}^{b}\geq z^{d}$
imply that the optimal policy prescribes that search ends with buying
$k$ before discovering $j$. In the latter, $w_{j}=\tilde{w}_{j}<w_{k}=\tilde{w}_{k}$
, and the optimal policy prescribes to continue discovering such that
both products are in the awareness set. The eventual purchase theorem
of \citet{Choi2016a} then applies, and hence the optimal policy does
not prescribe to buy product $j$. Second, consider the case where
$k$ is discovered after $j$ $\left(h_{k}>h_{j}\right)$. In this
case, note that $w_{j}>w_{k}$ if $\tilde{w}_{j}\geq z^{d}$. Hence,
$w_{k}>w_{j}$ if and only if $z^{d}>\tilde{w}_{k}>\tilde{w}_{j}$,
which is the same as (ii) above. Finally, consider the case where
$k$ is discovered at the same time as $j$ ($h_{k}=h_{j}$). Then
$w_{k}>w_{j}$ if and only if $\tilde{w}_{k}>\tilde{w}_{j}$, which
follows from the construction of the effective values. This again
is the same as (ii) above and hence the optimal policy does not prescribe
to buy $j$. 
\end{proof}

\subsection{Proposition \ref{prop:Expected-consumer-welfare} \label{subsec:proof-Expected-consumer-welfare}}
\begin{proof}
The proof follows a similar structure as the proof of Corollary 1
in \citet{Choi2016a}. To simplify exposition, the following additional
notation is used: Let $\tilde{w}_{j}\equiv x_{j}+\min\{y_{j},\xi_{j}\}$
as in Theorem \ref{theo:optimal_policy}, and $\hat{w}_{j}$ equal
to the effective value from Theorem \ref{theo:effective_values},
with the adjustment that $f(h_{j})=\varepsilon=0$. Furthermore, let
$\bar{w}_{r}\equiv\max_{k\in J_{0:r-1}}\hat{w}_{k}\forall r\geq1$,
$\tilde{\bar{w}}_{r}\equiv\max_{k\in J_{r}}\tilde{w}_{k}$ and $\tilde{\bar{w}}_{r,j}\equiv\max_{k\in J_{r}\backslash j}\tilde{w}_{k}$
where $J_{a:b}$ denotes the set of products discovered on position
$r\in\{a,\dots,b\}$, and $J_{r}$ is short-hand for $J_{r:r}$. Finally,
let $1(\cdot)$ denote the indicator function and $\bar{h}$ the maximum
position. 

The payoff of a consumer given realizations $x_{j}$ and $y_{j}$
for all $j$ is given by {\small{}
\begin{multline}
\sum_{r=1}^{\bar{h}}1(\bar{w}_{r}<z^{d})\left[\sum_{j\in J_{r}}1(\tilde{w}_{j}\geq\max\left\{ z^{d},\tilde{\bar{w}}_{r,j}\right\} )(x_{j}+y_{j})-1(x_{j}+\xi_{j}\geq\max\left\{ z^{d},\tilde{\bar{w}}_{r,j}\right\} )c_{s}\right]\\
+1(\bar{w}_{0}\geq z^{d})\nu_{0}-\sum_{r=1}^{\bar{h}}1(\bar{w}_{r}<z^{d})c_{d}+1(w_{\bar{h}}<z^{d})\nu
\end{multline} } 
which follows from the optimal policy and Theorem \ref{theo:effective_values}:
(i) If $\bar{w}_{0}\geq z^{d}$, the stopping rule implies that the
consumer does not discover any products beyond the initial awareness
set. Conditional on not discovering any additional products, the payoff
then is equal to $v_{0}$, which denotes the payoff of a directed
search problem over products $k\in S_{0}$ and an outside option offering
$\bar{u}_{0}=\max_{k\in C_{0}}u_{k}$. (ii) If $\bar{w}_{r}<z^{d}$,
the continuation rule implies that the consumer continues beyond position
$r-1$, i.e. discovers products on position $r$ and pays discovery
costs $c_{d}$. (iii) Conditional on discovering $j$, when $\tilde{w}_{j}\geq\max\left\{ z^{d},\tilde{\bar{w}}_{r,j}\right\} $,
the stopping and inspection rules imply that the consumer buys $j$,
gets utility $x_{j}+y_{j}$ and does not continue beyond position
$r$. (iv) Conditional on discovering $j$, when $x_{j}+\xi_{j}\geq\max\left\{ z^{d},\tilde{\bar{w}}_{r,j}\right\} $,
the inspection rule implies that the consumer inspects $j$ and incurs
costs $c_{s}$. (v) If $w_{\bar{h}}<z^{d}$, the continuation rule
implies that the consumer discovers all products, whereas the inspection
rule implies that he inspects all products $\left\{ j|x_{j}+\xi_{j}\geq z^{d}\right\} $.
Conditional having discovered all products, the consumer therefore
has the payoff of a directed search problem over products $\left\{ j|x_{j}+\xi_{j}<z^{d}\right\} $
with outside option $\tilde{u}_{0}=\max\{u_{0},\max_{k\in\left\{ j|x_{j}+\xi_{j}\geq z^{d},x_{j}+y_{j}\leq\xi_{j}\right\} }x_{k}+y_{k}\}$.
This is denoted by $\nu$.

Let $\mathbb{E}\left[\cdot\right]$ integrate over the distribution
of $X_{j},Y_{j}\forall j\in J$, and substitute inspection and discovery
costs by $c_{s}=\mathbb{E}\left[1(Y_{j}\geq\xi_{j})(Y_{j}+x_{j}-z_{j}^{s})\right]=\forall j$
(note that $z_{j}^{s}=x_{j}+\xi_{j}$) and $c_{d}=\mathbb{E}\left[1(\tilde{\bar{W}}_{r}\geq z^{d})(\tilde{\bar{W}}_{r}-z^{d})\right]$
(see Appendix \ref{sec:Details-on-the}). The expected payoff then
is given by:
\begin{align*}
& \hspace{1em}-1(X_{j}+\xi_{j}\geq\max\{z^{d},\tilde{\bar{W}}_{r,j}\})1(Y_{j}\geq\xi_{j})(Y_{j}-\xi_{j})\Biggr)\Biggr]\\
& \hspace{1em}-\sum_{r=1}^{\bar{h}}\mathbb{E}\left[1(\bar{W}_{r}<z^{d})1(\tilde{\bar{W}}_{r}\geq z^{d})(\tilde{\bar{W}}_{r}-z^{d})\right]+\mathbb{E}\left[1(\bar{W}_{0}\geq z^{d})\nu_{0}+1(\bar{W}_{\bar{h}}<z^{d})\nu\right]\\
= & \sum_{r=1}^{\bar{h}}\mathbb{E}\left[1(\bar{W}_{r}<z^{d})\Biggl(\sum_{j\in J_{r}}1(\tilde{W}_{j}\geq\max\{z^{d},\tilde{\bar{W}}_{r,j}\})(X_{j}+\min\{\xi_{j},Y_j\})\Biggr)\right]\\
& \hspace{1em}-\sum_{r=1}^{\bar{h}}\mathbb{E}\left[1(\bar{W}_{r}<z^{d})1(\tilde{\bar{W}}_{r}\geq z^{d})(\tilde{\bar{W}}_{r}-z^{d})\right]+\mathbb{E}\left[1(\bar{W}_{0}\geq z^{d})\nu_{0}+1(\bar{W}_{\bar{h}}<z^{d})\nu\right]\\
= & \sum_{r=1}^{\bar{h}}\mathbb{E}\left[1(\bar{W}_{r}<z^{d})1(\tilde{\bar{W}}_{r}\geq z^{d})\tilde{\bar{W}}_{r}\right]\\
& \hspace{1em}-\sum_{r=1}^{\bar{h}}\mathbb{E}\left[1(\bar{W}_{r}<z^{d})1(\tilde{\bar{W}}_{r}\geq z^{d})(\tilde{\bar{W}}_{r}-z^{d})\right]+\mathbb{E}\left[1(\bar{W}_{0}\geq z^{d})\nu_{0}+1(\bar{W}_{\bar{h}}<z^{d})\nu\right]\\
= & \sum_{r=1}^{\bar{h}}\mathbb{E}\left[1(\bar{W}_{r}<z^{d})1(\tilde{\bar{W}}_{r}\geq z^{d})z^{d}\right]\\
& \hspace{1em}+\mathbb{E}\left[1(\bar{W}_{0}\geq z^{d})\max\left\{ \bar{u}_{0},\max_{k\in S_{0}}\tilde{W}_{k}\right\} +1(\bar{W}<z^{d})\max\{\tilde{u}_{0},\max_{k\in\left\{ k|x_{k}+\xi_{k}<z^{d}\right\} }\tilde{W}_{k}\}\right]\\
& =\mathbb{E}\left[\max_{j\in J}\hat{W}_{j}\right]
\end{align*}
The second-to-last step substitutes $\nu_{0}=\mathbb{E}\left[\max\left\{ \bar{u}_{0},\max_{k\in S_{0}}\tilde{W}_{k}\right\} \right]$
and similarly for $\nu$, which directly follows from Corollary 1
in \citet{Choi2016a}. The last step combines the expressions of the
three mutually exclusive cases using the definition of $\hat{w}_{j}$.

To prove the second claim, note that the definition of $z^{d}$ requires
that $\mathbb{P}(\tilde{W}_{j}>z^{d})>0$, as otherwise $Q_{d}(c_{d},c_{s},z^{d})>0$.
Hence with $\left|J\right|=\infty$, $\mathbb{P}\left(\max_{j\in J}\tilde{W}_{j}<z^{d}\right)=0$
such that $\mathbb{E}\left[\max_{j\in J}\hat{W}_{j}\right]=z^{d}$. 
\end{proof}

\subsection{Proposition \ref{prop:welfare-change}} \label{subsec:proof-welfare-change} 
\begin{proof}
Consider a situation where we decrease costs $c_{s}$ and $c_{d}$
to either $c_{s}'=c_{s}-\Delta$ or $c_{d}'=c_{d}-\Delta$, while
keeping the other cost constant. Let $H_{1}(\cdot)$ and $H_{2}(\cdot)$
denote the cumulative density of $\bar{W}\equiv\max\{\bar{w}_{0},\max_{j\in J\backslash C_{0}\cup S_{0}}\hat{W}_{j}\}$
in the former and the latter case respectively, where $\bar{w}_{0}\equiv\max\left\{ \max_{k\in C_{0}}u_{k},\max_{k\in S_{0}}\tilde{w}_{k}\right\} $
is the value of the alternatives in the initial consideration and
awareness sets. Similarly, let $z_{1}^{d}$ and $z_{2}^{d}$ denote
the associated discovery values. Given $n_{d}=1$, we have $\frac{\partial Q_{d}(c_{d},c_{s},z)}{\partial c_{d}}<\frac{\partial Q_{d}(c_{d},c_{s},z)}{\partial c_{s}}$;
hence $\left|\frac{\partial z^{d}}{\partial c_{d}}\right|>\left|\frac{\partial z^{d}}{\partial c_{s}}\right|$
and $z_{2}^{d}>z_{1}^{d}$. Moreover, note that the definition of
the adjusted effective value $\hat{w}_{j}$ implies $H_{i}(w)=1\forall w\geq z_{i}^{d}$
and $H_{i}(w)=0\forall w\leq\bar{w}_{0}$. 

Conditional on $\bar{w}_{0}<z_{1}^{d}$, the difference in a consumer's
expected payoff across the two changes therefore can be written as
\begin{align}
\int_{z_{1}^{d}}^{z_{2}^{d}}1-H_{2}(w)\text{d}w- & \int_{\bar{w}_{0}}^{z_{1}^{d}}H_{2}(w)-H_{1}(w)\text{d}w
\end{align}
Whereas the first part is strictly positive, the second part is negative.
The latter follows as for $w\in[\bar{w}_{0},z_{1}^{d}]$, $\bar{W}=\max_{j\in J\backslash C_{0}\cup S_{0}}X_{j}+\min\{Y_{j},\xi\}$
and $\frac{\partial\xi}{\partial c_{s}}<0$ such that $H_{1}(w)\leq H_{2}(w)$.
As valuations are independent across products, we have $H_{k}(w)=\mathbb{P}_{X,Y}\left(X+\min\left\{ Y,\xi_{k}\right\} \leq w\right)^{\left|J\right|}$;
hence, as $\left|J\right|$ increases, $H_{2}(w)-H_{1}(w)$ and $H_{2}(w)$
decrease for $w\in[\bar{w}_{0},z_{2}^{d}]$.\footnote{Note that if $\mathbb{P}_{X,Y}\left(X+\min\left\{ Y,\xi_{k}\right\} \leq w\right)$
is large, then $H_{1}(w)-H_{2}(w)$ will first increase in $\left|J\right|$,
before starting to decrease. } Consequently, for all $\Delta>0$ there exists some threshold $n^{*}$
for $\left|J\right|$ such that the difference in the expected payoff
conditional on $\bar{w}_{0}<z_{1}^{d}$ is positive, i.e.
\begin{equation}
\int_{z_{1}^{d}}^{z_{2}^{d}}1-H_{2}(w)\text{d}w>\int_{\bar{w}_{0}}^{z_{1}^{d}}H_{2}(w)-H_{1}(w)\text{d}w\label{eq:threshold_appendix}
\end{equation}

Conditional on $\bar{w}_{0}\geq z_{1}^{d}$, having $z_{2}^{d}>z_{1}^{d}$
immediately implies that the expected payoff increases by at least
as much when decreasing discovery costs. Note also that when $z_{2}^{d}<\bar{w}_{0}$,
neither change affects the expected payoff. Finally, integrating over
the realizations $y_{k}$ for $k\in S_{0}$ that determine $\bar{w}_{0}$
yields the unconditional expected payoff as a combination of these
cases, which implies the first result. 

Increasing the value of the alternatives in the initial consideration
and awareness set then makes larger values of $\bar{w}_{0}$ more
likely. This implies the second result, as it makes both the case
$\bar{w_{0}}\geq z_{1}^{d}$ more likely, as well as decrease the
right-hand-side of \eqref{eq:threshold_appendix}. 
\end{proof}

\subsection{Proposition \ref{prop: diff-RS-SD stopping}} \label{subsec:proof-diff-rs-sd-stopping}
\begin{proof}
At $c_{s}=0$, we have $z^{d}=z^{RS}$. $\left|\frac{\partial z^{RS}}{\partial c_{s}}\right|\geq\left|\frac{\partial z^{d}}{\partial c_{s}}\right|$
then implies $z^{d}\geq z^{RS}$. Using this in \eqref{eq:stop_SD}
and \eqref{eq:stop_RS} immediately yields the result. 
\end{proof}

\subsection{Proposition \ref{prop: ranking-sd-rs}\label{subsec:proof-ranking-sd-rs}}
\begin{proof}
The first two statements immediately follow from \eqref{eq:ranking-eff-sd}
and \eqref{eq:ranking-eff-rs}. To see the latter, rewrite \eqref{eq:ranking-eff-sd}
as $\mathbb{P}_{W}\left(W<z^{d}\right)^{h-1}\mathbb{P}_{W}\left(W\geq z^{d}\right)^{2}$,
and \eqref{eq:ranking-eff-rs} in a similar way. $c^{RS}=c_{s}+c_{d}$
then implies $z^{d}\geq z^{RS}$. Hence, $\mathbb{P}_{W}\left(W<z^{d}\right)=\mathbb{P}_{X,Y}(X+\min\left\{ Y,\xi\right\} <z^{d})\geq\mathbb{P}_{X,Y}\left(X+Y<z^{RS}\right)$
which directly implies the existence of the threshold.
\end{proof}

\subsection{Proposition \ref{prop:ranking-ds}\label{subsec:ranking-ds}}
\begin{proof}
Write the first expression in \eqref{eq:ranking-eff-ds} (demand at
position $h$) as $\mathbb{E}_{\tilde{W}_{h}}\left[\mathbb{P}\left(\tilde{W}_{h+1}\leq\tilde{W}_{h}\right)\right.$
$\left.\prod_{k\notin\left\{ h,h+1\right\} }\mathbb{P}\left(\tilde{W}_{k}\leq\tilde{W}_{h}\right)\right]$.
When $\left|J\right|$ decreases, this expression decreases through
the product term, which is weighted by the first term $\mathbb{P}\left(\tilde{W}_{h+1}\leq\tilde{W}_{h}\right)$.
As $\mathbb{P}\left(\tilde{W}_{h+1}\leq t\right)\geq\mathbb{P}\left(\tilde{W}_{h}\leq t\right)\forall t$,
the first expression in \eqref{eq:ranking-eff-ds} decreases by more
than the second one when the number of alternatives increases.
\end{proof}

\subsection{Proposition \ref{prop:sd-rs-diff-exp-payoff}\label{subsec:proof-sd-rs-diff-exp-payoff}}
\begin{proof}
The RS problem is equivalent to a policy in the SD problem that commits
on inspecting every product that is discovered, conditional on which
the consumer chooses to stop optimally. However, as the optimal policy
in the SD problem is not this policy, it must yield a (weakly) larger
payoff. 
\end{proof}

\subsection{Uniqueness of discovery value}
\begin{proposition}
	\label{prop:unique_exp} \eqref{eq:res_exp} has a unique solution.
\end{proposition}
\begin{proof} $Q_d(c_{d},c_{s},z)$ with respect to $z$ yields
(see Appendix \ref{sec:Details-on-the})
\begin{equation}
\frac{\partial Q(c_{d},c_{s},z)}{\partial z}=\begin{cases}
+H(z)-1 & \text{if }z<0\\
-2+H(z) & \text{else}
\end{cases}
\end{equation}
where $H(\cdot)$ denotes the cumulative density of the random variable
$\max_{k\in\tilde{J}}\tilde{W}_{k}$. This implies $\frac{\partial Q_{d}(c_{d},c_{s},z)}{\partial z}\leq0$,
which combined with continuity, $Q_{d}(c_{d},c_{s},\infty)=-c_{d}$
and $Q_{d}(c_{d},c_{s},-\infty)=\infty$ imply that a solution to
\eqref{eq:res_exp} exists. Finally, uniqueness requires $Q_{d}(c_{d},c_{s},z)$
to be strictly decreasing at $z=z^{d}$. $\frac{\partial Q_{d}(c_{d},c_{s},z^{d})}{\partial z}=0$
would require that $H(z^{d})=1$, which contradicts the definition
of the discovery value value $z^{d}$ in \eqref{eq:res_exp}, as it
implies $Q_{d}(c_{d},c_{s},z^{d})\leq-c_{d}<0$.
\end{proof}	

\section{Further details on Search and Discovery Values \label{sec:Details-on-the}}

The search value of a product $j$ is defined by equation \eqref{eq:res_search}
and sets the myopic net gain of the inspection over immediately taking
a hypothetical outside option offering utility $z$ to zero. This
myopic net gain can be calculated as follows:\footnote{The second steps holds as with a change in the order of integration we get $\int_{z-x_j}^\infty[1-F(y)]\mathrm{d}y = \int_{z-x_{j}}\int_{y}^{\infty}f_{Y}(t)\mathrm{d}t\mathrm{d}y=\int_{z-x_{j}}\int_{z-x_{j}}^{t}f_{Y}(t)\mathrm{d}y\mathrm{d}t =\int_{z-x_{j}}\left[yf_{Y}(t)\right]_{y=z-x_{j}}^{y=t}\mathrm{d}t$.}

\begin{align*}
Q_{s}(x_{j},c_{s},z) & =\mathbb{E}_{Y}\left[\max\{0,x_{j}+Y-z\}\right]-c_{s}\\
 & =\int_{z-x_{j}}^{\infty}(x_{j}+y-z)\mathrm{d}F(y)-c_{s}\\
 & =\int_{z-x_{j}}^{\infty}\left[1-F(y)\right]\mathrm{d}y-c_{s}
\end{align*}
Substituting $\xi_{j}=z-x_{j}$ then yields \eqref{eq:res_search_explicit}. 

The discovery value is defined by equation \eqref{eq:res_exp} and
sets the expected myopic net gain of discovering more products over
immediately taking a hypothetical outside option offering utility
$z$ to zero. Corollary 1 in \citet{Choi2016a} and similar steps
as the above then imply that: 
\begin{align*}
Q_{d}(c_{d},c_{s},z) & =\mathbb{E}_{\boldsymbol{X},\boldsymbol{Y}}\left[\max\left\{ z,\max_{k\in\left\{ 1,\dots,n_{d}\right\} }\tilde{W}_{k}\right\} \right]-z-c_{d}\\
 & =\mathbb{E}_{\boldsymbol{X},\boldsymbol{Y}}\left[\max\left\{ 0,\max_{k\in\left\{ 1,\dots,n_{d}\right\} }\tilde{W}_{k}-z\right\} \right]-c_{d}\\
 & =\int_{z}^{\infty}1-H(w)\mathrm{d}w-c_{d}
\end{align*}
where $H(\cdot)$ denotes the cumulative density of the random variable
$\max_{k\in\tilde{J}}\tilde{W}_{j}$. The above also implies that
in the case where $Y$ is independent of $X$, a change in variables yields that the discovery value is linear in the
mean of $X$, denoted by $\mu_{X}$:
\[
z^{d}=\mu_{X}+\Xi(c_{s},c_{d})
\]
 where $\Xi(c_{s},c_{d})$ solves \eqref{eq:res_exp} for an alternative
random variable $\tilde{X}=X-\mu_{X}$. 

\section{Monotonicity and Extensions\label{sec:Generalizations}}

Monotonicity of the Gittins index values ($g_{t}^{d}\geq g_{t+1}^{d}\forall t$)
is satisfied whenever the following holds:

\begin{align}
0\leq & \mathbb{E}_{\boldsymbol{X},Y,n_{d},q,t}\left[\Theta(\Omega_{t+1},\tilde{A}_{t+1},g_{t}^{d})\right]\nonumber \\
 & -\mathbb{E}_{\boldsymbol{X},Y,n_{d},q,t+1}\left[\Theta(\Omega_{t+2},\tilde{A}_{t+2},g_{t+1}^{d})\right]\label{eq:condition_mono}
\end{align}
where $g_{t}^{d}$ is the Gittins index of discovering products (defined
by \eqref{eq:def_gittins}), and $\tilde{A}_{t+1}\equiv\{d,s1,\dots,sn_{d}\}$
is the set of actions available in $t+1$ containing the newly revealed
products and (if available) the possible future discoveries. The expectation
operator $\mathbb{E}_{\boldsymbol{X},Y,n_{d},J,t}\left[\cdot\right]$
integrates over the following random realizations, where the respective
joint distribution now can be time-dependent: (i) Partial valuations
drawn from $\boldsymbol{X}=\left[X_{1},\dots,X_{n_{d}}\right]$; (ii)
conditional distributions $F_{Y|X=x}(y)$; (iii) the number of revealed
alternatives ($n_{d}$); (iv) whether more products can be discovered
in future periods determined by the belief $q$.

It goes beyond the scope of this paper to determine all possible specifications
of beliefs which satisfy this condition. However, Proposition
\ref{prop:monotone_specifications} provides two specifications that
can be of interest and for which \eqref{eq:condition_mono} holds
(see also Section \ref{sec:Optimal-Search}).
\begin{proposition}
\label{prop:monotone_specifications} \eqref{eq:condition_mono} holds
for the below deviations from the baseline model:
\begin{enumerate}
\item $Y$ is independent of $X$. Beliefs are such that the revealed partial
valuations in $\boldsymbol{X}$ are i.i.d. with time-dependent cumulative
density $G_{t}(x)$ such that $G_{t}\left(x\right)\leq G_{t+1}\left(x\right)\forall x\geq z^{d}-\xi$.
\item The consumer does not know how many alternatives he will discover.
Instead, he has beliefs such that with each discovery, at most the
same number of alternatives are revealed as in previous periods ($n_{d,t+1}\leq n_{d,t}$).
\end{enumerate}
\end{proposition}
\begin{proof}
Each part is proven using slightly different arguments.
\begin{enumerate}
\item Let $\tilde{x}\equiv\max_{k\in\{1,\dots,n_{d}\}}x_{k}$. If $\tilde{z}^{s}$$=\tilde{x}+\xi\leq z^{d}$,
$\Theta(\Omega_{t+1},\tilde{A}_{t+1},z^{d})=1$, whereas for $\tilde{x}>z^{d}-\xi$,
$\frac{\partial\Theta(\Omega_{t+1},\{e,s1,\dots,sn_{d}\},z^{d})}{\partial\tilde{x}}\geq0$.
Independence implies that the cumulative density of the maximum $\tilde{x}$
is $\tilde{G}_{t}(x)=G_{t}(x)^{n_{d}}$. Consequently, whenever the
distribution of $X$ shifts such that $G_{t}(x)\leq G_{t+1}(x)\forall x\geq z^{d}-\xi$,
larger values of $\Theta(\Omega_{t+1},\tilde{A}_{t+1},g_{t}^{d})$
become less likely in $t+1$, and hence \eqref{eq:condition_mono}
holds.
\item Since $\frac{\partial\Theta\left(\Omega_{t+1},\left\{ s_{k}\right\} ,w\right)}{\partial w}\leq1$,
we have $\frac{\partial\Theta(\Omega_{t+1},\tilde{A}_{t+1},g_{t}^{d})}{\partial n_{d}}\geq0$.
Hence \eqref{eq:condition_mono} holds given $n_{d,t+1}\leq n_{d,t}$.
\end{enumerate}
\end{proof}
Based on this monotonicity condition, Proposition \ref{prop:generalization_monotonicity}
generalizes Theorem \ref{theo:optimal_policy}. It implies that whenever
\eqref{eq:condition_mono} holds, the discovery value can be calculated
based on the expected myopic net gain of discovering products over
immediately taking the hypothetical outside option. Hence, whenever
\eqref{eq:condition_mono} holds, the optimal policy continues to
be fully characterized by reservation values that can be obtained
without having to consider many future periods.
\begin{proposition}
\label{prop:generalization_monotonicity}Whenever \eqref{eq:condition_mono}
is satisfied, Theorem \ref{theo:optimal_policy} continues to hold
(with appropriate adjustment of the discovery value's time-dependence).
\end{proposition}
\begin{proof}
Follows directly from the proof of Theorem \ref{theo:optimal_policy}.
\end{proof}

\section{Violations of independence assumption\label{sec:Violations-of-independence}}

\textbf{Costly recall:} Consider a variation to the search problem,
where purchasing a product in the consideration set is costly unless
it is bought immediately after it is inspected. If in period $t$
product $j$ is inspected, then inspecting another product or discovering
more products in $t+1$ will change the payoff of purchasing product
$j$ by adding the purchase cost. In the context of a multi-armed
bandit problem, this case arises if there are nonzero costs of switching
between arms. \citet{Banks1994}, for example, provide a more general
discussion on switching costs and the nonexistence of optimal index-based
strategies. The same reasoning also applies in a search problem where
inspecting a product is more costly if the consumer first discovers
more products. The exception is if there are infinitely many alternatives.
In this case, the optimal policy never prescribes to recall an alternative.

\textbf{Learning:}\emph{ }Independence is also violated for some types
of learning. Consider a variation of the search problem, where the
consumer updates his beliefs on the distribution of $Y$. In this
case, by inspecting a product $k$ and revealing $y_{ik}$, the consumer
will update his belief about the distribution of $Y$, thus affecting
the expected payoffs of both discovering more and inspecting other
products. Independence therefore is violated and the reservation value
policy is no longer optimal.\footnote{\citet{Adam2001} studies a similar case where independence continues
to hold across groups of products. However, his results do not extend
to the case with limited awareness, as the beliefs of $Y$ also determine
the expected benefits of discovering more products.} Note, however, that as long as learning is such that only payoffs
of actions that will be available in the future are affected, independence
continues to hold. This is for example the case when the consumer
learns about the distribution of $X$ as discussed in Section 3.2.

\textbf{Purchase without inspection:}\emph{ }A final setting where
independence does not hold is when a consumer can buy a product without
first inspecting it. In this case, the consumer has two actions available
for each product he is aware of. He can either inspect a product,
or directly purchase it. Clearly, when the consumer first inspects
the product, the information revealed changes the payoff of buying
the product. Independence therefore is violated and the reservation
value policy is not guaranteed to be optimal. \citet{Doval2018} studies
this search problem for the case where a consumer is aware of all
available alternatives, and characterizes the optimal policy under
additional conditions. 

\section{Learning\label{sec:Learning}}

Several studies consider priors or learning rules under which the
optimal policy is myopic when searching with recall \citep{Rothschild1974,Rosenfield1981,Bikhchandani1996,Adam2001}.
A sufficient condition for the optimal policy to be myopic is given
in Theorem 1 of \citet{Rosenfield1981}: Once the expected net benefits
of continuing search over stopping with the current best option are
negative, they remain so. Hence, whenever it is optimal to stop in
$t$, it is also optimal to stop in all future periods. The monotonicity
condition used in this paper directly imposes that this is satisfied;
expected benefits of discovering more products remain constant or
decrease during search. A fairly general assumption underlying learning
rules that satisfy this condition is Assumption 1 in \citet{Bikhchandani1996}.
This assumption requires that beliefs are updated such that values
above the largest value revealed so far become less likely.\footnote{Note that \citet{Bikhchandani1996} consider search for low prices.}
Hence, whenever a better value is found than the current best, finding
an even better match in the future becomes less likely. 

In the SD problem, similar learning rules that satisfy this condition
are difficult to find. When the consumer learns about the number of
products that are revealed with each discovery, expected benefits
of discovering more products increase if many products are revealed,
but the value of stopping remains the same if all these products are
bad matches. Hence, a learning rule would need to guarantee that beliefs
shift such that the expected benefits of discovering more products
do not increase, as opposed to only the net benefits over stopping.
Similarly, if the consumer learns about the distribution of partial
valuations $\boldsymbol{X}$, the value of stopping need not increase
even if partial information indicates a good match leading the consumer
to shift beliefs towards larger values; after inspecting a promising
product, the consumer may still realize that the product is worse
than the previously best option. 

Though the optimal policy is not myopic with learning, it is still
based on the Gittins index, where the search and purchase values are
as in the baseline SD problem. The main difficulty is calculating
the index value for discovering more products, denoted by $z_{t}^{L}$.
Whereas calculating this value precisely would require accounting
for learning in future periods, it is possible to derive bounds on
this value that are easier to calculate and can be used to judge how
far off a myopic policy is. 

To show this, I focus on the case where
the consumer learns about the distribution of partial valuations.
In particular, consider the following variation of the search and
discovery problem: Let the distribution of partial valuations in $\boldsymbol{X}$
be characterized by a parameter vector $\theta$ and denote its cumulative
density by $G_{\theta}(\cdot)$. The consumer initially does not know
the true parameter vector and Bayesian updates his beliefs in period
$t$ given some prior distribution. Denoting the consumers' beliefs
on $\theta$ with cumulative density $P_{t}(\cdot)$, the consumers'
beliefs about $\boldsymbol{X}$ drawn in the next discovery are characterized
by the cumulative density $\tilde{G}_{t}(\boldsymbol{X})=\int G_{\theta}(\boldsymbol{X})\text{d}P_{t}(\theta)$.\footnote{For example, consider the case of sampling from a Normal distribution
with unknown mean and known variance, and assume $n_{d}=1$. If the
consumer believes in $t$ that the mean is distributed normally with
$\theta\sim N(\mu_{t},\sigma_{t}^{2})$, then $\tilde{G}_{t}(x)=\Phi(\frac{x-\mu_{t}}{\sigma_{t}})$,
where $\Phi(\cdot)$ is the standard normal cumulative density \citep[see e.g. Theorem 1 in][Ch. 9.5]{DeGroot1970}.} 

Denote a \emph{k}-step look-ahead value as $z_{t}^{d}(k)$ and define
it as the value of a hypothetical outside option that makes the consumer
indifferent between stopping immediately, and discovering more products
after which at most $k-1$ more discoveries remain. For example, $z_{t}^{d}(1)$
satisfies the myopic comparison in \eqref{eq:res_exp}, where expectations
are calculated based on period $t$ beliefs $\tilde{G}_{t}(\cdot)$.
The definition of $z_{t}^{d}(1)$ then implies that it is equal to
the expected value of continuing to discover products if no future
discoveries remain. As the consumer can stop and take this hypothetical
outside option in $t+1$, allowing for more discoveries after $t+1$
can only increase the expected value, hence $z_{t}^{d}(1)\leq z_{t}^{d}(2)\dots\leq z_{t}^{L}$.
$z_{t}^{d}(1)$ therefore provides a lower bound on $z_{t}^{L}$,
and $z_{t}^{L}$ can be approximated with increasing precision through
\emph{k}-step look-ahead values. 

To derive an upper bound, consider the case where the consumer learns
the true $\theta$ in $t+1$, if he chooses to discover more products
in $t$. The value of discovering more products in $t$ when the true
$\theta$ is revealed in $t+1$ then is larger compared to the case
where the consumer continues to learn. This is formally derived by
\citet{State1970} for a search problem where a consumer samples from
an unknown distribution. Intuitively, when the true $\theta$ is revealed,
the consumer is able choose the action in $t+1$ that maximizes the
expected payoff going forward for each realization of $\theta$. In
contrast, if the consumer does not learn the true $\theta$ in $t+1$,
he cannot choose the maximizing action for each realization of $\theta$,
but only the action that maximizes expected payoff on average across
possible $\theta$. 

An upper bound therefore is given by the value $\bar{z}_{t}^{d}$
such that the consumer is indifferent between stopping and taking
a hypothetical outside option offering $\bar{z}_{t}^{d}$, and discovering
more products after which the true $\theta$ is revealed. Formally,
$\bar{z}_{t}^{d}$ satisfies 
\begin{equation}
\bar{z}_{t}^{d}=\int\int\tilde{V}(\Omega_{t+1},A_{t+1},\bar{z}_{t};\theta)\text{d}P_{t+1}(\theta)\text{d}\tilde{G}_{t}(\boldsymbol{X})\label{eq:upper_bound}
\end{equation}
where $\tilde{V}(\Omega_{t+1},A_{t+1},\bar{z}_{t}^{d};\theta)$ denotes
the expected value of a search and discovery problem with known $\theta$
and an outside option offering $\bar{z}_{t}^{d}$. Proposition
\ref{prop:Expected-consumer-welfare} then directly allows to calculate this value
without having to consider all the possible search paths.

Proposition \ref{prop:bounds} summarizes these results. A similar
result can also be derived for the case where the consumer learns
about a distribution from which the number of products that are discovered
is drawn. 
\begin{proposition}
\label{prop:bounds} In the search and discovery problem with Bayesian
learning about an unknown distribution of partial valuations $\boldsymbol{X}$,
it is optimal to: 
\begin{enumerate}
\item continue whenever $\max_{k\in C_{t}}u_{k}\leq z_{t}^{d}(1)$
\item stop whenever $\max_{k\in C_{t}}u_{k}\geq\bar{z}_{t}^{d}$ 
\end{enumerate}
\end{proposition}

\section{Estimation Details\label{sec:Estimation-Details}}

To estimate the three models I use a simulated maximum likelihood
approach based on a kernel-smoothed frequency simulator. Using numerical
optimization, parameters are found that maximize the simulated likelihood
given by: 
\[
\max_{\gamma}\sum_{i}\mathcal{L}_{i}(\gamma)=\sum_{i}\log\left(\frac{1}{N_{d}}\sum_{d=1}^{N_{d}}\frac{1}{1+\sum_{k=1}^{N_{k}}\exp(-\lambda\kappa_{kdi})}\right)
\]
where $\gamma$ is the parameter vector, $N_{d}$ is the number of
simulation draws, $\lambda$ is a smoothing parameter and $\kappa_{kd}$
is one of $N_{k}$ inequalities resulting from the optimal policy
in the respective model evaluated for draw $d$. All three models
are estimated with $\lambda=10$ and $N_{d}=500$. At these values,
parameters are recovered well when data is generated with the same
model. 

\textbf{DS conditions} These conditions are the same as in \citet{Ursu2018},
who provides further details on how they relate the optimal policy
in the DS problem. The difference to her specification is that inspection
costs are linear, and that in DS1 there are no positions. For observed consideration
set $C_{i}$ for consumer $i$, a given draw $d$ for the unobserved
taste shocks $y_{j}(d)$ which defines product utilities $u_{j}(d)$
as well as the utility of the purchased option $u_{i}^{*}(d)$, there
are multiple purchase, and stopping conditions expressed in inequalities:
\begin{align*}
 & \text{Stopping:} & \kappa_{kdi} & =\max_{j\in C_{i}}u_{j}(d)-z_{m}\forall m\notin C_{i}\\
 & \text{Continuation} & \kappa_{kdi} & =z_{m+1}-\max_{j\in C_{i}(m)}u_{j}(d)\forall m=1,2,\dots,N_{is}-1\\
\text{} & \text{Purchase:} & \kappa_{kdi} & =u_{i}^{*}(d)-u_{j}(d)\forall j\in C_{i}
\end{align*}
In the continuation conditions, $N_{is}$ denotes the number of observed
inspections, $z_{m+1}$ is the search value of the next inspection,
 $C_{i}(m)$ is the consideration set of $i$ after $m$ inspections.
Note that the last relies on observing the order in which products
are inspected; if this order were not observed,
the method proposed by \citet{Honka2017} could be used to integrate
over possible search orders. The stopping condition only applies if
not all products are inspected, the continuation condition only applies
if $i$ inspected at least one product.

\textbf{RS conditions }The conditions in the RS model are similar
to the ones in the DS model. However, the stopping and continuation
conditions now are based on the reservation value $z^{RS}$, which
follows directly from the optimal policy: 
\begin{align*}
 & \text{Stopping:} & \kappa_{kdi} & =\max_{j\in C_{i}}u_{j}(d)-z^{RS}\\
 & \text{Continuation} & \kappa_{kdi} & =z^{RS}-\max_{j\in C_{i}(m)}u_{j}(d)\forall m=1,2,\dots,N_{is}-1\\
\text{} & \text{Purchase:} & \kappa_{kdi} & =u_{i}^{*}(d)-u_{j}(d)\forall j\in C_{i}
\end{align*}

\textbf{FI conditions }In the FI model, standard purchase conditions
apply: 
\begin{align*}
\kappa_{kdi} & =u_{i}^{*}(d)-u_{j}(d)\forall j
\end{align*}

\section{Sellers' decisions\label{sec:Sellers'-decisions}}

To illustrate the difference in sellers' decision making across the
SD and DS problem, we can compare the market demand generated by the
SD problem with the one from the DS problem when there are infinitely
many alternatives. Given a unit mass of consumers, market demand for
a product discovered at position $h$ is given by 

\begin{align}
d_{SD}(h) & =\mathbb{P}_{\boldsymbol{W}}\left(W_{k}<z^{d}\forall k<h\right)\mathbb{P}_{W_{h}}\left(W_{h}\geq z^{d}\right)
\end{align}
where $W_{h}$ is the random effective value of a product on position
$h$. The expression immediately follows from the stopping decision
which implies that if a consumer discovers a product with $w_{j}\geq z^{d}$,
he will stop searching and buy a product $j$. Hence, the consumer
will only discover and have the option to buy a product on position
$h$ if $w_{h}<z^{d}$ for all products on earlier positions. 

For the DS problem, \citet{Choi2016a} showed that the market demand
is given by
\begin{equation}
d_{DS}(h)=\mathbb{P}_{\boldsymbol{W}}\left(\tilde{W}_{h}\geq\max_{k\in J}\tilde{W}_{k}\right)
\end{equation}
where $\tilde{W}_{k}=X_{k}+\min\left\{ Y_{k},\xi_{k}\right\} $. 

Now suppose that the seller of a product on position $h$ sets the
mean of $X_{h}$, for example by choosing a price. In the SD problem,
this is equivalent to choosing $\mathbb{P}_{W_{h}}\left(W_{h}\geq z^{d}\right)$;
the probability that the consumer inspects and then stops search by
buying the seller's product. Importantly, this does not directly depend
on partial valuations of both products at earlier, and products at
later positions. This results from the stopping decisions, and given
the infinite number of products a consumer will never recall a product
discovered earlier. 

In contrast, in the DS problem, choosing the mean of $X_{h}$ influences
demand through the joint distribution of all products. As consumers
are aware of all products, they compare all partial valuations. Hence,
each seller's choice of partial valuations affects all other sellers
demand, and sellers do not make independent decisions. 

\end{document}